\documentclass{llncs}
\usepackage{amssymb}
\usepackage{xspace}
\usepackage[english]{babel}
\usepackage{amsmath}
\usepackage{amsfonts}
\usepackage{latexsym}
\usepackage{subfigure}
\usepackage{color}
\usepackage{enumerate}
\usepackage{tabularx}
\usepackage{todonotes}

\graphicspath{{figures/}}

\newcommand{\remove}[1]{}
\newcommand{\ap}{\textsc{Arrow-Placement}\xspace}
\newcommand{\dap}{\textsc{Discrete-Arrow-Placement}\xspace}
\newcommand{\ov}[1]{\operatorname{ov}(#1)}
\newcommand{\opt}{\textsf{Opt}\xspace}
\newcommand{\heurglobal}{\textsf{HeurGlobal}\xspace}
\newcommand{\heurlocal}{\textsf{HeurLocal}\xspace}
\newcommand{\editor}{\textsf{Editor}\xspace}
\newcommand{\planar}{\textsc{Planar}\xspace}
\newcommand{\random}{\textsc{Random}\xspace}
\newcommand{\north}{\textsc{North}\xspace}

\title{Placing Arrows in Directed Graph Drawings% 
\thanks{Work is partially supported by the MIUR project AMANDA ``Algorithmics for MAssive and Networked DAta'', prot. 2012C4E3KT\_001. }
}

\author{
Carla Binucci\inst{1},
Markus Chimani\inst{2},
Walter Didimo\inst{1},\\
Giuseppe Liotta\inst{1},
Fabrizio Montecchiani\inst{1}
}

\date{}

\institute{
Universit\`a degli Studi di Perugia, Italy\\
\email{\{carla.binucci,walter.didimo,giuseppe.liotta,fabrizio.montecchiani\}@unipg.it} 
\and
Osnabr{\"u}ck University, Germany,
\email{markus.chimani@uni-osnabrueck.de}
}

\begin{document}
\maketitle

\begin{abstract}
We consider the problem of placing arrow heads in directed graph drawings without them overlapping other drawn objects. This gives drawings where edge directions can be deduced unambiguously. We show hardness of the problem, present exact and heuristic algorithms, and report on a practical study.
\end{abstract}

%-------------
% INTRODUCTION
%-------------
\section{Introduction}\label{se:introduction}

The default way of drawing a directed edge is to draw it as a line with an arrow head at its target. While there also exist other models (placing arrows at the middle, drawing edges in a ``tapered'' fashion, etc.; cf.~\cite{DBLP:conf/apvis/HoltenIWF11,Holten:2009:USV:1518701.1519054}) the former is prevailing in virtually all software systems. However, this simple model becomes problematic when several edges attach to a vertex on a similar trajectory: it may be hard to see whether a specific edge is in- or outgoing, cf.\ Fig.~\ref{fi:example} and Figs.~\ref{fi:large-example-North} and~\ref{fi:large-example-Planar-Random} in the appendix.

We try to solve this issue by looking for a placement of the arrow heads such that (a) they do not overlap other edges or arrow heads, and (b) still retain the property of being at---or at least close to---the target vertices of the edges. In the following, we show NP-hardness of the problem, propose exact and heuristic algorithms for its discretized variant, and evaluate their practical performance in a brief exploratory study. We remark that our problem is related to map labeling and in particular to edge labeling problems~\cite{DBLP:conf/isaac/GemsaNN13,DBLP:conf/wea/GemsaNR14,DBLP:journals/comgeo/KakoulisT01,DBLP:reference/crc/KakoulisT13,DBLP:journals/comgeo/KreveldSW99,Marks91thecomputational,DBLP:journals/geoinformatica/StrijkK02,DBLP:journals/ijcga/StrijkW01,DBLP:journals/algorithmica/WagnerWKS01,Wolff00asimple}. 

For space reasons, some proofs and technical details are omitted in this extended abstract, and can be found in the appendix.
\begin{figure}
\centering
\subfigure{\includegraphics[scale=0.3]{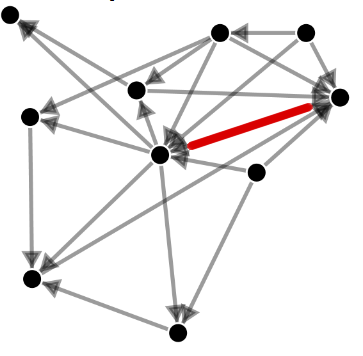} }\hfil
\subfigure{\includegraphics[scale=0.3]{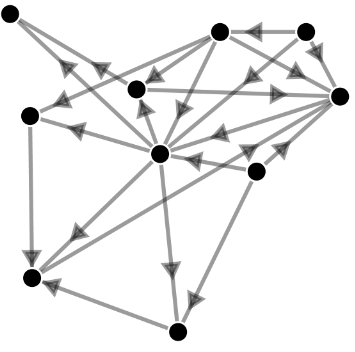}}
\caption{\label{fi:example}Layouts of a digraph with 10 vertices and 21 edges. (left) The arrows are placed by a common editor; several arrows overlap and the direction of, e.g., the thick red edge is not clear. (right) The arrows are placed by our exact method. }
\end{figure}

\section{The Arrow Placement Problem}\label{se:problem}
We first formally define our arrow placement problem and establish its theoretical time complexity.
%In this section we formally define our arrow placement problem and establish its theoretical time complexity. We start with some notations and preliminary definitions.
%
Let $G=(V,E)$ be a digraph and let $\Gamma$ be a straight-line drawing of $G$. We assume that in $\Gamma$ each vertex $v \in V$ is drawn as a circle (possibly a point) $C_v$. We also assume that, for each edge $e \in E$, the arrow of $e$ is modeled as a circle $C_e$ of positive radius, centered in a point along the segment that represents $e$: when $\Gamma$ is displayed, the arrow of $e$ is drawn as a triangle inscribed in $C_e$, suitably rotated according to the direction of $e$.
%We assume that the circles representing the vertices have all the same radius, which we denote by $r_V$; analogously, the circles modeling the arrows have all the same radius, which we denote by $r_E$. 
We assume that all circles representing a vertex (arrow) have a common radius $r_V$ ($r_E$, respectively).
We say that two arrows---or an arrow and a vertex---\emph{overlap} if their corresponding circles intersect in two points.
An arrow and an edge \emph{overlap} if the segment representing the edge intersects the circle representing the arrow in two points. For the sake of simplicity, we reuse terms of theoretical concepts also for their visual representation: ``arrow'' and ``vertex'' also refer to their corresponding circle in $\Gamma$; ``edge'' also refers to its corresponding segment in $\Gamma$.

\begin{definition}\label{de:valid-placement}
Let $a_e$ denote the arrow of an edge $e\in E$. A \emph{valid position} for $a_e$ in $\Gamma$ is such that: \texttt{(P1)} for every vertex $v \in V$, $a_e$ and $v$ do not overlap; \texttt{(P2)} for every edge $g \in E$, $g\neq e$, $a_e$ and $g$ do not overlap. 
An assignment of a valid position to each arrow is called a \emph{valid placement} of the arrows, denoted by $P_\Gamma$.
\end{definition}

\begin{definition}\label{de:overlap-number}
Given a valid placement $P_\Gamma$, the \emph{overlap number} of $P_\Gamma$ is the number of pairs of overlapping arrows, and is denoted as $\ov{P_\Gamma}$.  
\end{definition}

Given a straight-line drawing $\Gamma$ of a digraph $G=(V,E)$, and constants $r_V$, $r_E$, we ask for a valid placement $P_\Gamma$ of the arrows (if one exists) such that $\ov{P_\Gamma}$ is minimum. This optimization problem is NP-hard; we prove this by showing the hardness of the following decision problem \ap.

\medskip\noindent{\bf Problem:} \ap

\smallskip\noindent\textsc{Instance:} $\langle G=(V,E), \Gamma, r_V, r_E \rangle$.

\smallskip\noindent\textsc{Question:} Does there exist a valid placement $P_\Gamma$ of the arrows with $\ov{P_\Gamma}=0$?

\begin{theorem}\label{th:np-hardness}
The \ap problem is NP-hard.
\end{theorem}

%Given a circle $C$, we denote by $c(C)$ the center of $C$. Also, given two points $p$ and $q$, the Euclidean distance between $p$ and $q$ is denoted by $d(p,q)$.
%For every edge $g \in E$ distinct from $e$, $a_e$ and $a_g$ do not overlap;
%In particular, if $e=(u,v)$, \texttt{(P2)} implies that $d(c(C_e),c(C_w)) \geq r_V + r_E$, $\;\;$ $w \in \{u,v\}$.

The proof of Thm.~\ref{th:np-hardness} uses a reduction from \textsc{Planar 3-SAT}~\cite{DBLP:journals/siamcomp/Lichtenstein82}, and is similar to those used in the context of edge labeling~\cite{DBLP:journals/comgeo/KakoulisT01,Marks91thecomputational,DBLP:journals/ijcga/StrijkW01,Wolff00asimple}. It yields an instance of \ap where the search of a valid placement $P_\Gamma$ with $\ov{P_\Gamma}=0$ can be restricted to a finite number of valid positions for each arrow. Hence, \ap remains NP-hard even if we fix a finite set of positions for each arrow, and a valid placement with overlap number zero (if any) may only choose from these positions. As this variant of \ap, which we call \dap, clearly belongs to NP, it is NP-complete.

\section{Algorithms}\label{se:algo}

We describe algorithms for the optimization version of \dap. We assume that a set of valid positions for each arrow is given, based on $\{\Gamma, r_V, r_E\}$, and look for a valid placement $P_\Gamma$ that minimizes $\ov{P_\Gamma}$ over this set of positions. We give both an exact algorithm and two variants of a heuristic, which we experimentally compare in Section~\ref{se:experiments}. 
Given an edge $e \in E$, let $A_e$ denote the set of valid positions for the arrow of edge $e$, and let $A:=\bigcup_{e \in E} A_e$ be the set of all valid positions. Our algorithms are based on an \emph{arrow conflict graph} $C_A$, depending on $A$, $\Gamma$, and $r_E$. The positions $A$ form the node set of $C_A$. Two positions are \emph{conflicting}, and connected by an (undirected) edge in $C_A$, if they correspond to positions of different edges and the arrows would overlap when placed on these positions. Finding a valid placement $P_\Gamma$ with $\ov{P_\Gamma}=0$ means to select one element from each $A_e$ such that they form an independent set in $C_A$. More general, finding a valid placement $P_\Gamma$ with $\ov{P_\Gamma}=k$ ($k \geq 0$) means to select one element from each $A_e$ such that they induce a subgraph with $k$ edges in  $C_A$. Our exact algorithm minimizes $k$ using an ILP formulation, while our heuristic adopts a greedy strategy. Both techniques try to minimize the distance of each arrow from its target vertex as a secondary objective. 
%This is motivated by the fact that many systems for displaying digraphs place the arrows close to the target vertices, thus these are the positions where the user may expect to see the arrows. %%% this is now already stated in the intro
However, our algorithms can be easily adapted to privilege other positions (e.g., close to the source vertices, in the middle of the edges, etc.), or to consider bidirected edges.

\paragraph{ILP formulation.} 
For each position $p_e\in A_e$ of an edge $e =(v,u)$, we have a binary variable $x_{p_e}$. We define a distance $d(p_e) \in \{1, \ldots, |A_e| \}$, from $p_e$ to $u$:  $d(p_e)=1$ ($d(p_e)=|A_e|$) means that $p_e$ is the position closest (farthest, respectively) to $u$. Let $E_A:=E(C_A)$ be the pairs of conflicting positions. For every $(p_e,p_g)\in E_A$, we define a binary variable $y_{p_e p_g}$. The total number of variables is $O(|A|^2)$, and we write:
%The objective function and the constraints of our ILP model are as follows.
\begin{align}
\min \sum_{(p_e,p_g)\in E_A} y_{p_e p_g} &+ \frac{1}{M} \cdot \sum_{e \in E}\sum_{p_e \in A_e} d(p_e) x_{p_e} \label{con:obj}\\
&\sum_{p_e \in A_e} x_{p_e} = 1 && \forall e \in E\label{con:oneLabelxEdge}\\
&x_{p_e} + x_{p_g} \leq y_{p_e p_g} + 1 && \forall (p_e,p_g) \in E_A\label{con:isConflict}
\end{align}

The objective function minimizes the overlap number and, secondly, the sum of the distances of the arrows from their target vertices. To do this, the second term is divided for a sufficiently large constant $M$. For example, one can set $M = |E|\max_{e \in E}\{|A_e|\}$.
Equations~\eqref{con:oneLabelxEdge} guarantee that exactly one valid position per edge is selected. Constraint~\eqref{con:isConflict} enforces $y_{p_e p_g}=1$ if both conflicting positions $x_{p_e}$ and $x_{p_g}$ are chosen.
%, namely when the corresponding conflicting positions are selected. Otherwise, if only one or none of the two positions is selected, Constraint~\ref{con:isConflict} together with the objective function~\ref{con:obj} guarantees that $y_{p_e p_g}$ is set equal to zero. 
In the following, the exact technique will be referred to as \opt.    
We remark that optimization problems and ILP formulations similar to above have been given in the context of edge and map labeling~\cite{DBLP:conf/isaac/GemsaNN13,DBLP:conf/wea/GemsaNR14,DBLP:journals/comgeo/KakoulisT01,DBLP:journals/comgeo/KreveldSW99,DBLP:journals/geoinformatica/StrijkK02,DBLP:journals/ijcga/StrijkW01}. 

\paragraph{Heuristics.} Our heuristics follow a greedy strategy, again based on $C_A$. Let $p_e\in A_e\subset V(C_A)$ as above. We initially assigns cost $c(p_e)$ to each position $p_e$, and then execute $|E|$ iterations. In each iteration, we select a position $p_e$ of minimum cost (over all $e\in E$) and place the arrow of the corresponding edge there; then, we remove all positions $A_e$ from $C_A$ (including $p_e$), and update the costs of the remaining positions.
We define $c(p_e) := \delta(p_e) + \frac{1}{M}  d(p_e) + T\sigma_{p_e}$, where: $\delta(p_e)$ is the degree of $p_e$ in $C_A$ (i.e., the number of positions conflicting with $p_e$); constant $M$ and ``distance'' $d(p_e)$ are defined as in the ILP; $\sigma_{p_e}$ is the number of already chosen positions conflicting with $p_e$ (0 in the first iteration); $T$ is 
equal to the maximum initial cost of a valid position.
This cost function guarantees that: $(i)$ positions conflicting with already selected positions are chosen only if necessary; $(ii)$ the algorithm prefers positions with the minimum number of conflicts with the remaining positions and, among them, those closer to the target vertex. 
Since constructing $C_A$ may be time-consuming in practice (we compare all pairs of valid positions), we also consider using only a subset of the edges of $C_A$; we may consider only those conflicts arising from positions of adjacent edges in the input graph. In the following, \heurglobal is the heuristic that considers full $C_A$, while \heurlocal is the variant based on this simplified version of~$C_A$.

\section{Experimental Analysis}\label{se:experiments}

%\begin{itemize}
%\item experimental goals (what we want to learn from these experiments, more than giving hypothesis)
%\item choices for generating the set $A$ of valid positions, including the fact that we insert one specific position in $A_e$ when $a_e$ has no valid position (for practical reasons).
%\item digraph benchmarks and structural information about the number of valid positions (total, avg per edge, number of edges with one non-valid position).
%\item measures: running time, quality metrics (arrow overlap, maximum/avg distance from the target vertices)
%\item machine characteristics (important to evaluate running time).
%\item presentation and discussion of the results (1 exact technique, 2 heuristics, 1 trivial placement).
%\end{itemize}

\paragraph{Experimental Setting.} We use three different sets of graph: 
\planar are biconnected planar digraphs with edge density $1.5$--$2.5$, randomly generated with the OGDF~\cite{OGDF}.
\random are digraphs generated with uniform probability distribution with edge density $1.4$--$1.6$. Both sets contain $30$ instances each; $6$ graphs for each number of vertices $n \in \{100, 200, \dots, 500\}$. We did not generate denser graphs, as they give rise to cluttered drawings with few valid positions for the arrows---there, the arrow placement problem seems less relevant. Finally, \north is a popular set of $1,275$ real-world digraphs with $10$--$100$ vertices and average density $1.4$~\cite{north}. We draw each instance of the three sets with straight-line edges using OGDF's FM3 algorithm~\cite{DBLP:conf/gd/HachulJ04}. The layouts of the \planar may contain edge crossings, as they are generated by a force-directed approach. 

Value $r_E$ is chosen as the minimum of $(a)$ $40\%$ of the shortest edge length, $(b)$ $25\%$ of the average edge length, and $(c)$ 10 pixels, but enforced to be at least 3 pixels. We set $r_V:=r_E$. For each edge $e=(w,u)$ we compute positions $A_e$ as follows. The $i$-th position, $i\geq 1$, has its center at distance $r_V + i \cdot r_E$ from target~$u$. We generate positions as long as they have distance at least $r_V+r_E$ from source vertex $w$. We then remove positions that overlap with edges or vertices in $\Gamma$.  If no valid positions remain, we choose the one closest to $u$ as $e$'s unique arrow position. Thus, in the final placements there might be some conflicts between an arrow and a vertex or edge of the drawing. We call such conflicts \emph{crossings} and observe that a single invalid position may result in several crossings.

 We apply \opt, \heurglobal, and \heurlocal to each of the drawings. 
 The algorithms are implemented in C\# and run on an Intel Core i7-3630QM notebook with $8$ GB RAM under Windows $10$. For the ILP we use %IBM/Ilog 
 CPLEX 12.6.1 with default settings. 
For each computation, we measure total running time, overlap number, and number of crossings (due to invalid positions, see above). From the qualitative point of view, we also compare the algorithmic results with a trivial placement, called \editor, which simply places each arrow close to its target vertex, similarly as most graph editors do. %Besides the total running time, 
We also measure \emph{placement time}, i.e., the time spent by an algorithm to find a placement \emph{after} $C_A$ has been computed.  

\newcommand{\mypicsize}{0.42\columnwidth}
\begin{figure}[p]
  \centering
  	\subfigure[]{
    \includegraphics[width=\mypicsize] {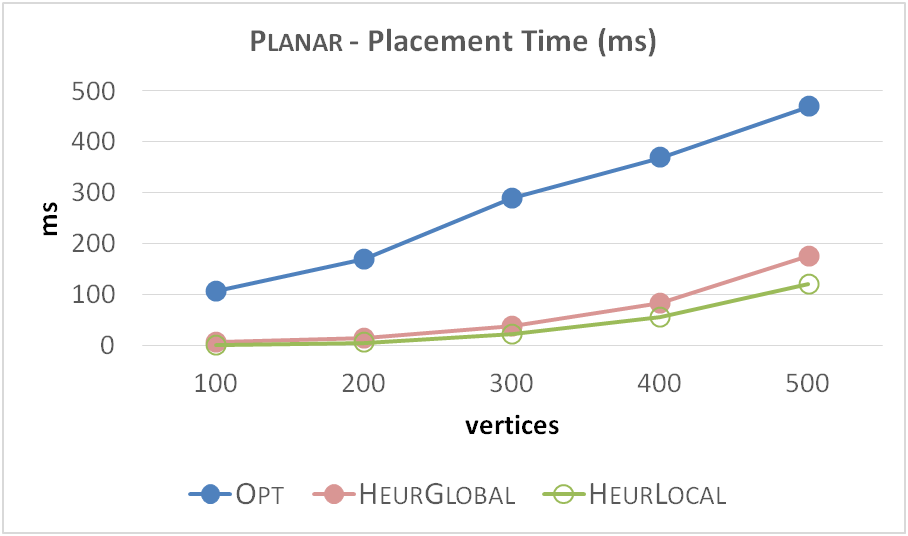}
    \label{fi:pl-solvingTime}
	} \hfill 
    \subfigure[]{
    \includegraphics[width=\mypicsize] {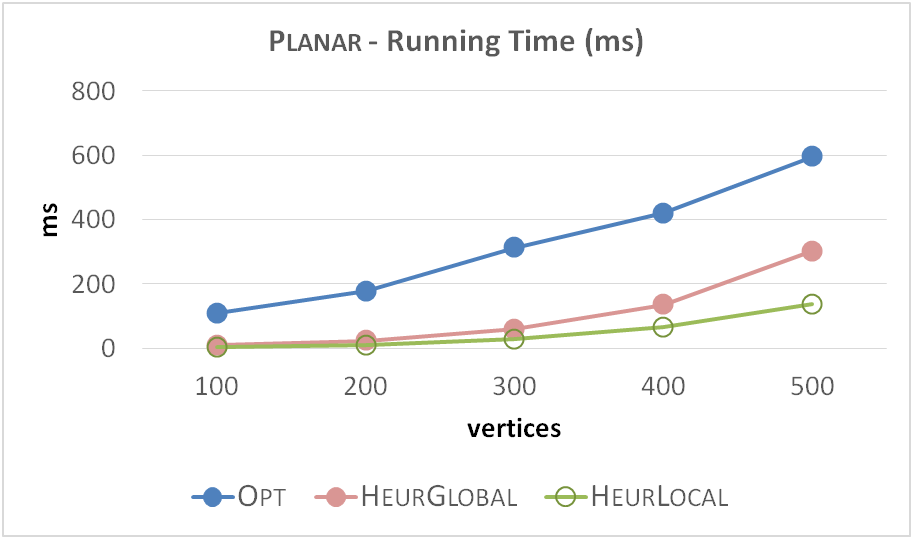}
    \label{fi:pl-solvingTime-ConflictGraph}
	} 
    \subfigure[]{
    \includegraphics[width=\mypicsize] {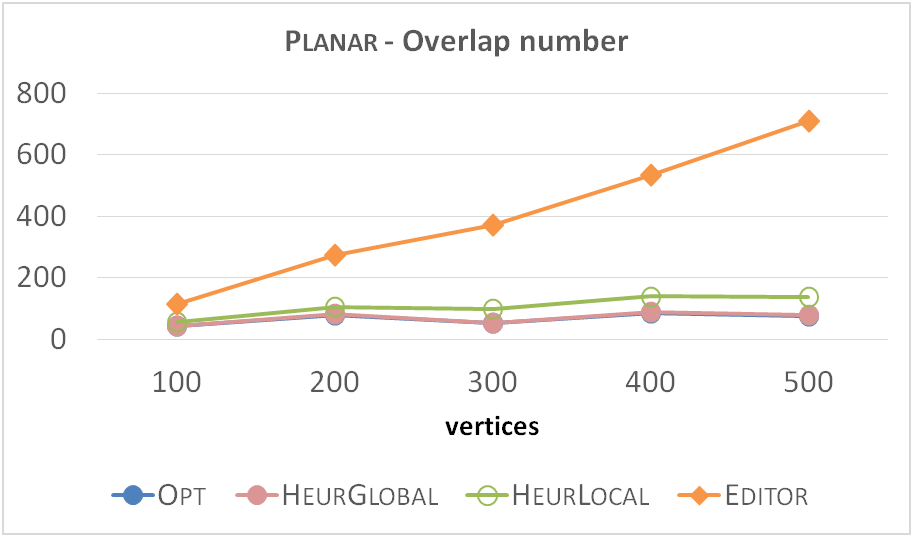}
    \label{fi:pl-conflicts}
	} \hfill 
   \subfigure[]{
    \includegraphics[width=\mypicsize] {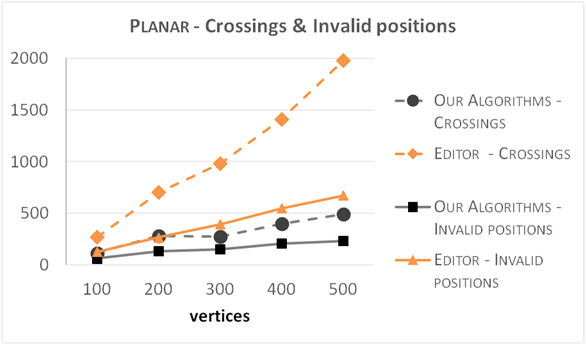}
    \label{fi:pl-crossings}
	} 
\caption{\planar: (a) placement time;
(b) total running time;
(c) number of overlaps;
(d) number of crossings (edge/vertex with arrow) and of invalid positions.}
\label{fi:planar}
\end{figure}
\begin{figure}[p]
  \centering
    \subfigure[]{
    \includegraphics[width=\mypicsize] {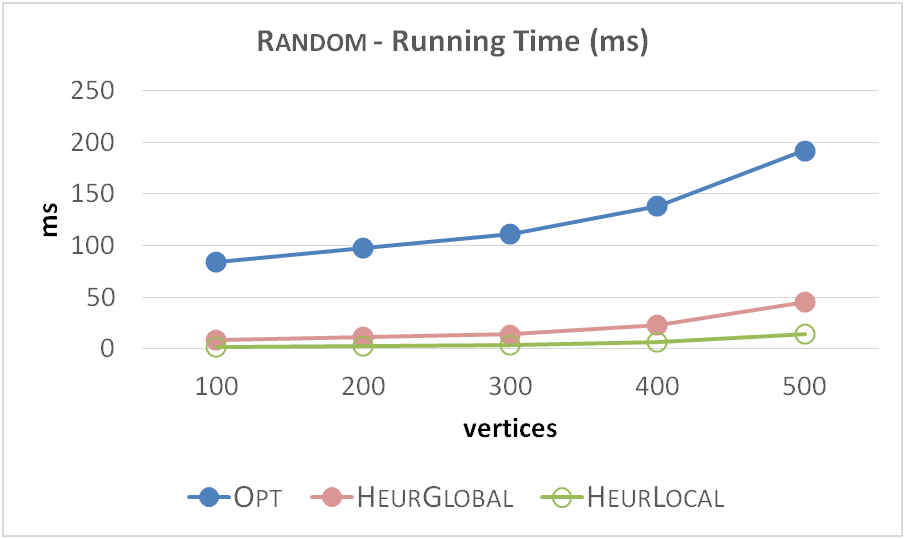}
    \label{fi:rand-solvingTime-ConflictGraph}
	} \hfill
    \subfigure[]{
    \includegraphics[width=\mypicsize] {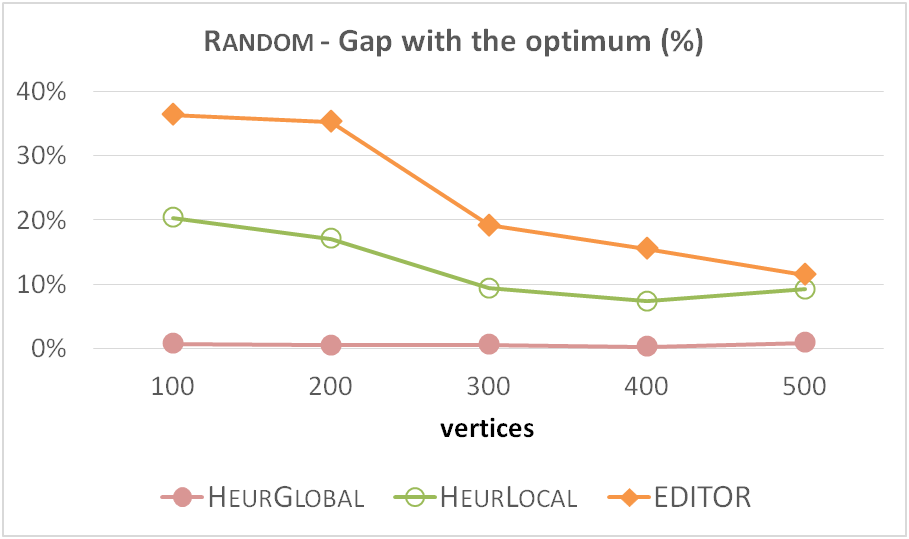}
    \label{fi:rand-conflicts-gap}
	} 
\caption{\random:
(a) total running time; 
(b) number of overlaps, relative to~\opt.}
\label{fi:random}
\end{figure}
\begin{figure}[p]
  \centering
    \subfigure[]{
    \includegraphics[width=\mypicsize] {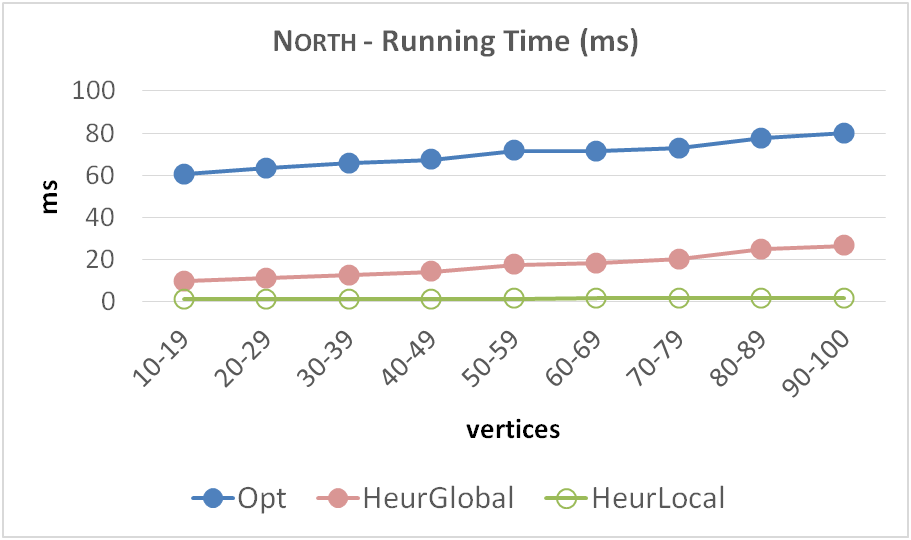}
    \label{fi:north-solvingTime-ConflictGraph}
	} \hfill
    \subfigure[]{
    \includegraphics[width=\mypicsize] {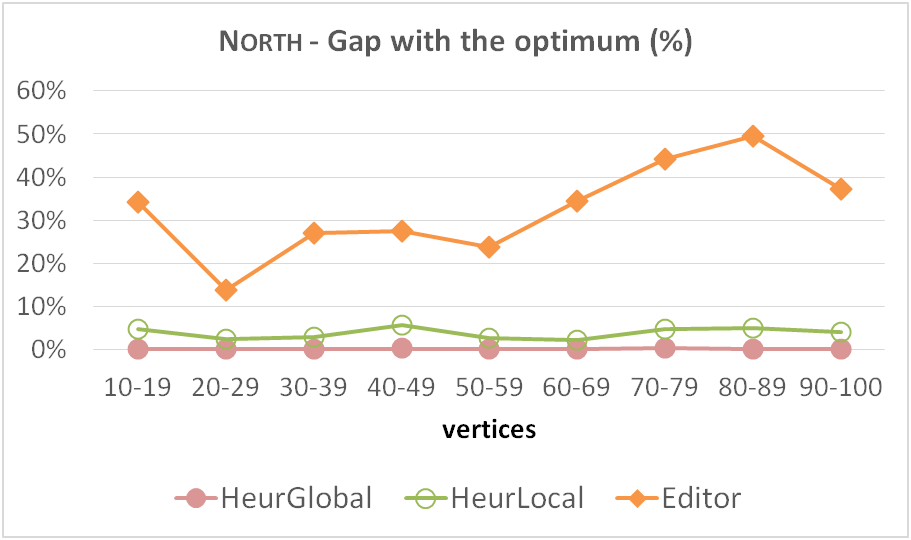}
    \label{fi:north-conflicts-gap}
	} 
\caption{\north:
(a) total running time; 
(b) number of overlaps, relative to to~\opt.}
\label{fi:north}
\end{figure}

\paragraph{Results.} 
For \planar, the average numbers of positions in $C_A$ range from $640$ to $7,150$. 
Figs.~\ref{fi:pl-solvingTime} and~\ref{fi:pl-solvingTime-ConflictGraph} show that for  \planar all the algorithms are very applicable, although \opt is of course significantly slower. While the pure placement time for \heurglobal is not much longer than that of \heurlocal, it suffers from the fact that generating the full $C_A$ constitutes roughly 1/3 of its overall runtime, whereas the generation time of the reduced conflict graph is rather neglectable.
On the other hand, Fig.~\ref{fi:pl-conflicts} shows that \heurglobal 
practically coincides with the optimum w.r.t.\ the number of overlaps (its average gap is below $3\%$; the worst gap is $6.76\%$). \heurlocal still gives very good solutions, with gaps about half that of \editor. 
Figure~\ref{fi:pl-crossings} shows that our algorithms reduce the number of invalid positions by $33$--$77\%$ compared to \editor.
The number of \emph{crossings} is the same for all our algorithms, as they occur when we cannot find any valid position for arrows during the generation procedure. Figure~\ref{fi:pl-crossings} shows that our algorithms cause significantly less crossings than \editor.

%XXXXXXX

For \random, average numbers of positions in $C_A$ range from $640$ to $4,377$. 
The general behavior for \random is similar to that of \planar but the difference between the running time of \opt and the heuristics is slightly more pronounced (Fig.~\ref{fi:rand-solvingTime-ConflictGraph}). Again, constructing $C_A$ constitutes roughly 1/3 of \heurglobal's running time.
Still, the quality of \heurglobal's solutions again essentially coincide with \opt; the other heuristics are now closer than before, see Fig.~\ref{fi:rand-conflicts-gap}.

For \north, the average $|V(C_A)|$ range from $62$ to $311$. We observe the same patterns, see Figs.~\ref{fi:north}: \heurlocal requires nearly no time, while \heurglobal is very competitive at just above 20ms for the large graphs (a third of which is the construction of full $C_A$). Again, \opt always finds a solution very quickly, in fact within roughly 80ms. \heurglobal again gives essentially optimal solutions, while \heurlocal exhibits $5$--$10\%$ gaps. \editor requires $30$--$50\%$ more overlaps than \opt.

%XXXXXXXXX

\section{Conclusions and Future Work}\label{se:conclusions}
We discussed optimizing arrow head placement in directed graph drawings, to improve readability. As mentioned, this is very related to studies in map and graph labeling, but its specifics seem to make a more focused study worthwhile.

Our techniques are of practical use, and could be sped-up by constructing $C_A$ using a sweepline or the labeling techniques in~\cite{DBLP:journals/algorithmica/WagnerWKS01}. It would be interesting to validate the effectiveness of our approach through a user study (e.g. for tasks that involve path recognition). Moreover, one may consider both placing labels and arrow heads.
Finally, the non-discretized problem variant, as well as the variants' respective (practical) benefits, should be investigated in more depth.

\paragraph{Acknowledgments.}
Research on this problem started at the Dagstuhl seminar 15052~\cite{brandes_et_al:DR:2015:5041}. We thank Michael Kaufmann and Dorothea Wagner for valuable discussions, and the anonymous referees for their comments and suggestions.

%{\color{red} MC: some more things in the experiments section:
%
%*) number of positions (number of nodes in $C_A$). we may show this either in the diagrams as gray vertical bars in the background (labeled on the right-hand y-axis). or we may tell in the text how many positions we have per edge on average. {\color{blue} We added a sentence in the text at the beginning of the results discussion. We also added two corresponding charts in the appendix.}
%
%*) NumberOfCrossings: can we integrate the number of invalid positions into these diagrams? (i.e. just counting each arrow once whether it is valid or not) -- if yes, then there is an out-commented sentence discussing it in the results-part.] {\color{blue} Done}
%
%*) Fig \planar-OverlapNumber: the line for \opt is close to invisible. we could truncate the y-axis at 400 (instead of 800) and just write the editor-numbers in the caption. {\color{blue} We tried but it was still invisible....}
%
%*) Fig \random-PlacementTime typo in diagram title (RandoM) {\color{blue} Done}
%
%*) Fig \random-OverlapNumber typo in diagram title (two spaces before Overlap) {\color{blue} Done}
%   
%*) we can remove the diagram titles, since the text is already in the latex figure caption. this allows to scale the diagrams a bit. {\color{blue} Not sure, titles are always useful in case of a quick view of the referee...}
%}

{\small \bibliography{arrows}}
\bibliographystyle{splncs03}

\newpage

\section*{Appendix A -- Additional examples}
Figure~\ref{fi:large-example-North} depicts some larger real-world layouts from our graph benchmark. 
Examples of layouts for larger artificial instances are shown in Fig.~\ref{fi:large-example-Planar-Random}. 
When the visual complexity of the graph layout increases, the arrow placement problem seems to become less relevant: indeed, cluttered drawings have few valid positions for the arrows and also the solutions of our algorithms may contain several crossings between arrows and other objects; see, e.g., Figs.~\ref{fi:large-Random-ED} and~\ref{fi:large-Random-EX}.

\begin{figure}[]
  \centering
  	\subfigure{
    \includegraphics[width=\mypicsize] {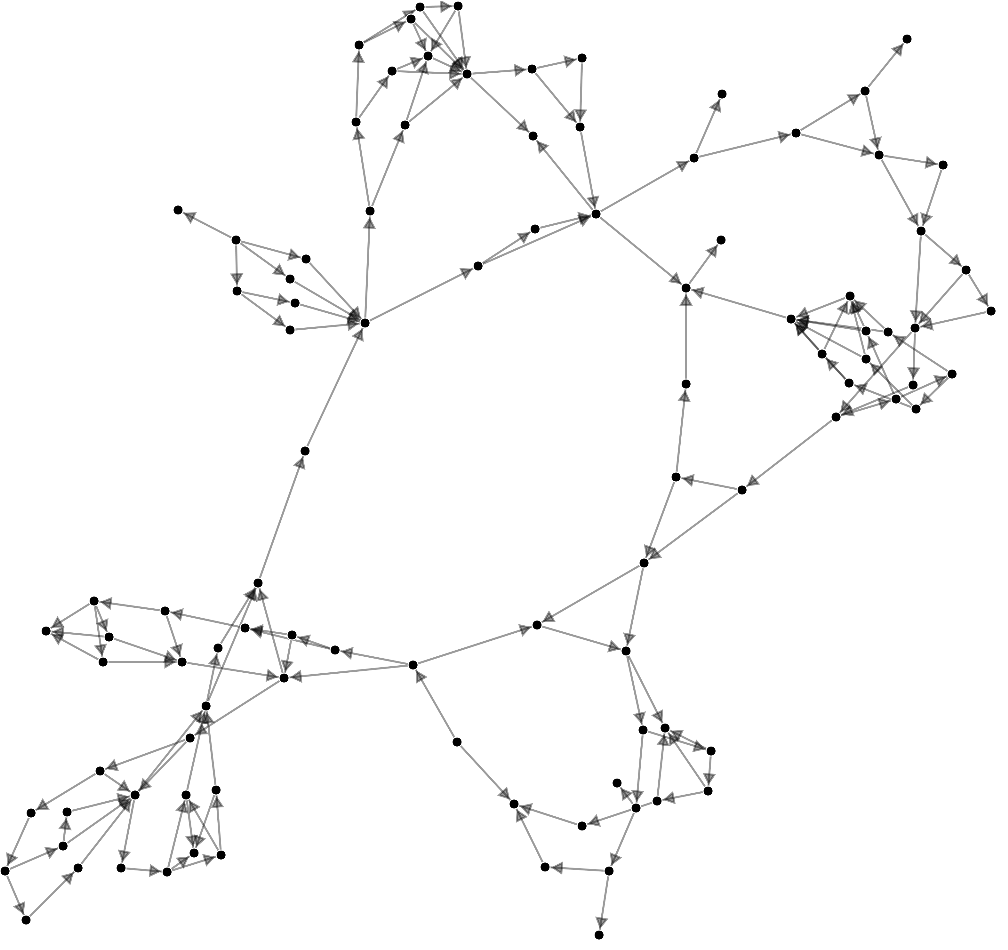}} \hfill
    \subfigure{
    \includegraphics[width=\mypicsize] {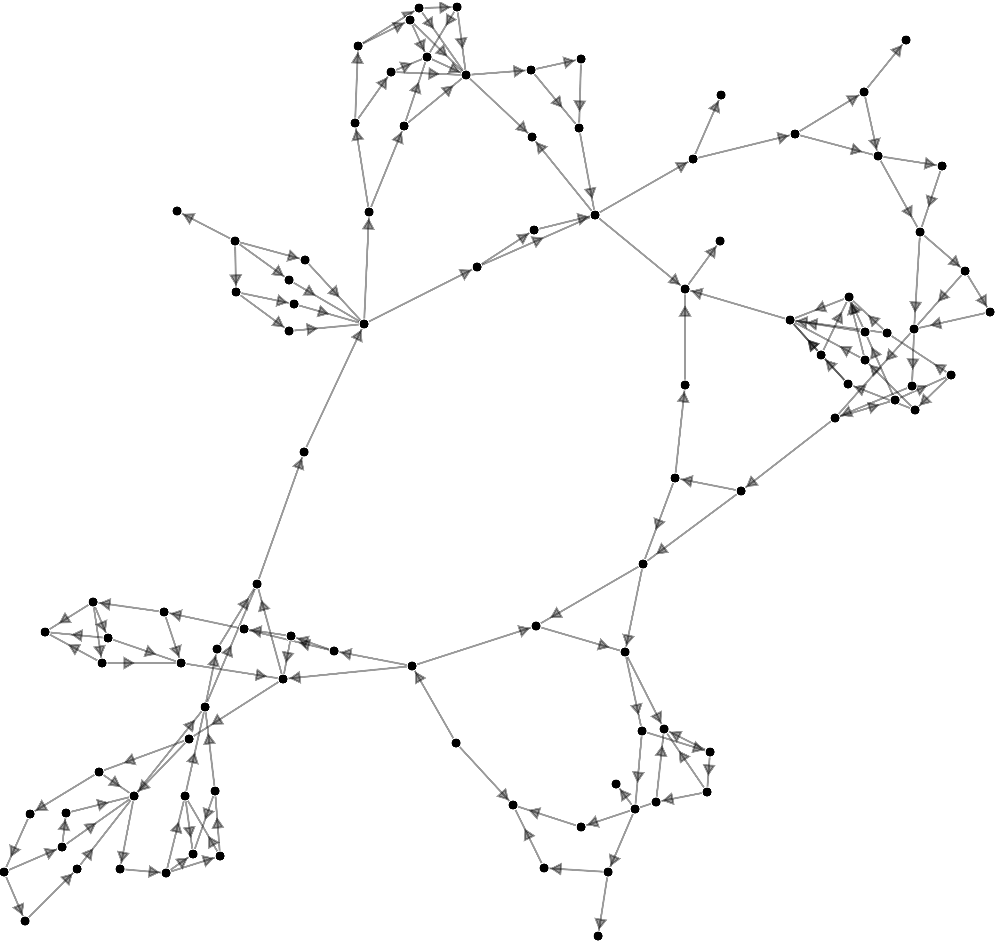}}
    \subfigure{
    \includegraphics[width=\mypicsize] {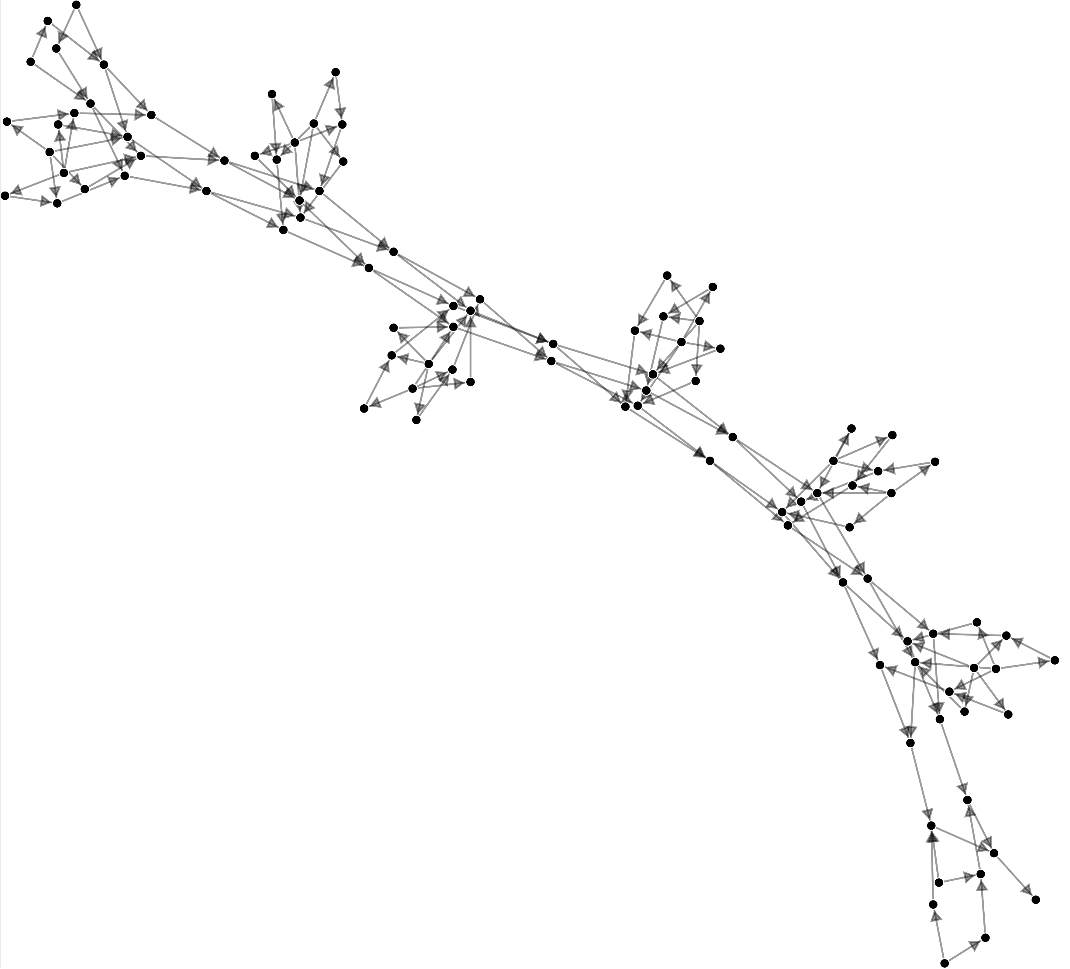}} \hfill
   \subfigure{
    \includegraphics[width=\mypicsize] {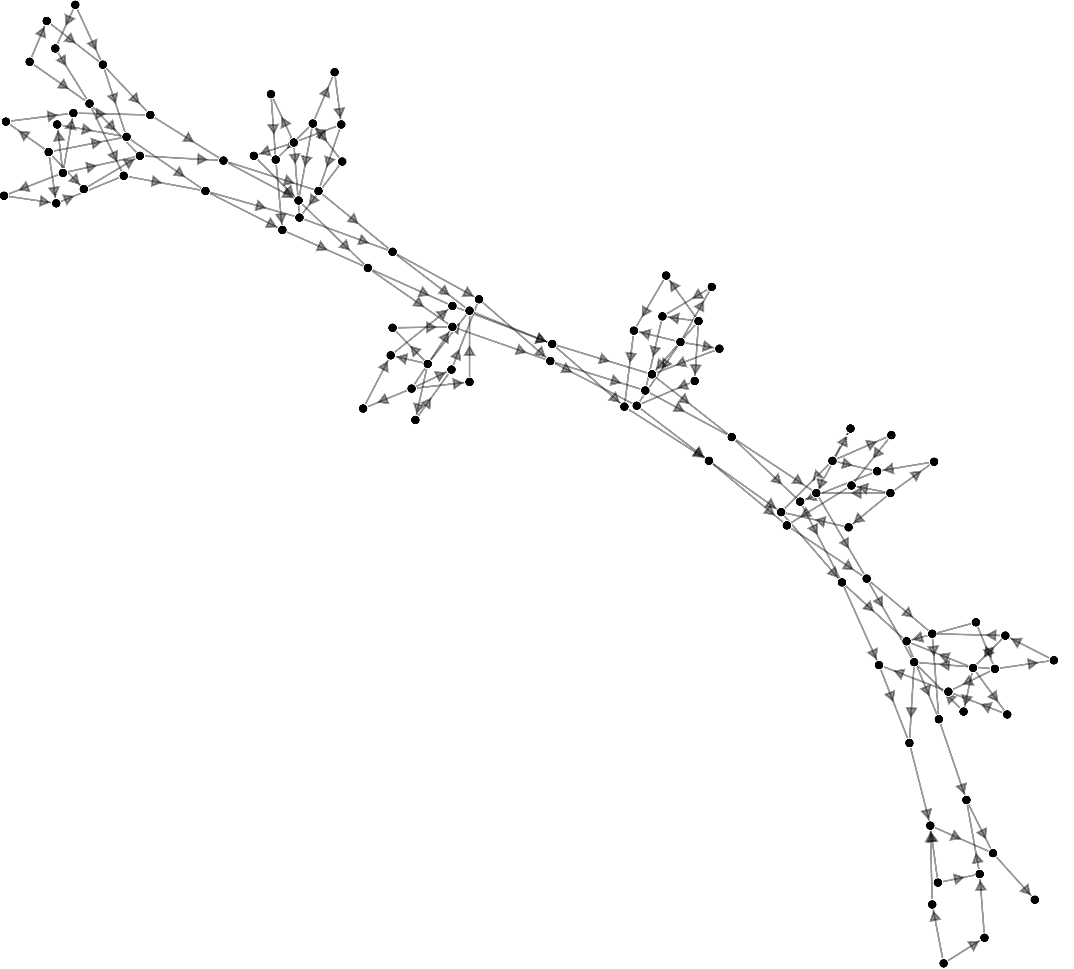}} 
\caption{\label{fi:large-example-North}Drawings of two \north digraphs with 99 vertices and 148 edges. On the left, arrows are placed by a common editor; several arrows overlap. On the right, arrows are placed by our exact technique. }
\end{figure}

\begin{figure}[!h]
  \centering
  	\subfigure[]{
    \includegraphics[width=\mypicsize] {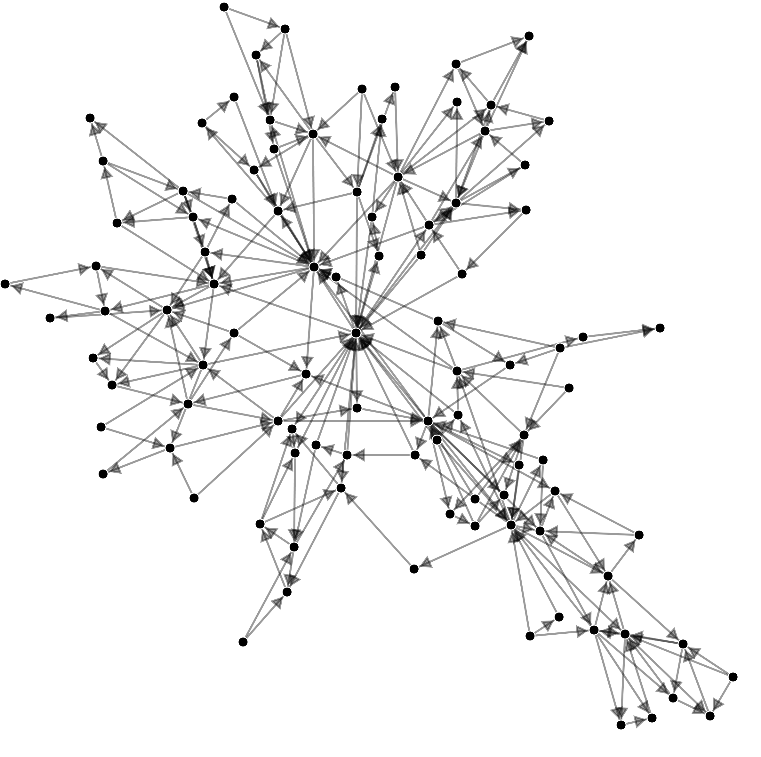}
    \label{fi:large-Planar-ED}
	} 
    \subfigure[]{
    \includegraphics[width=\mypicsize] {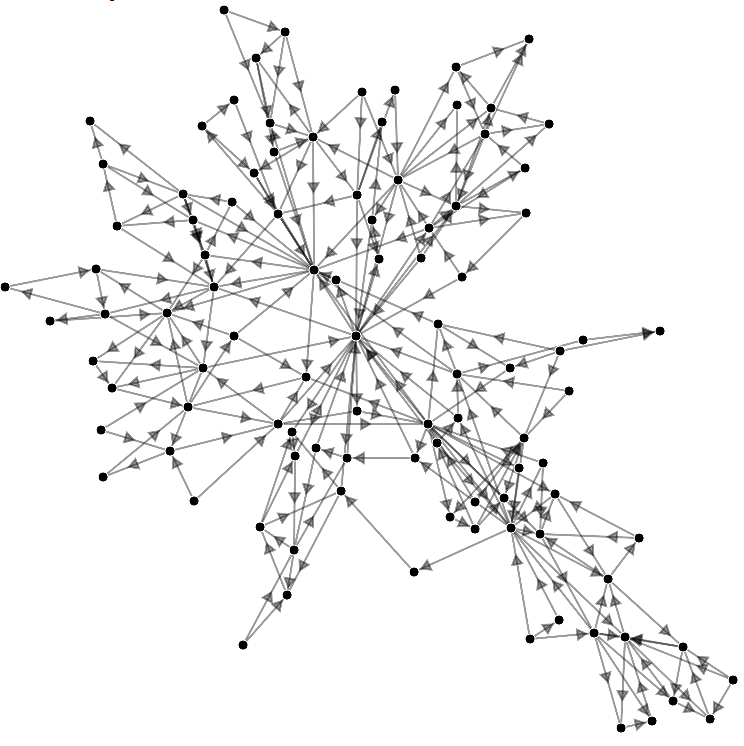}
    \label{fi:large-Planar-EX}
	} 
    \subfigure[]{
    \includegraphics[width=\mypicsize] {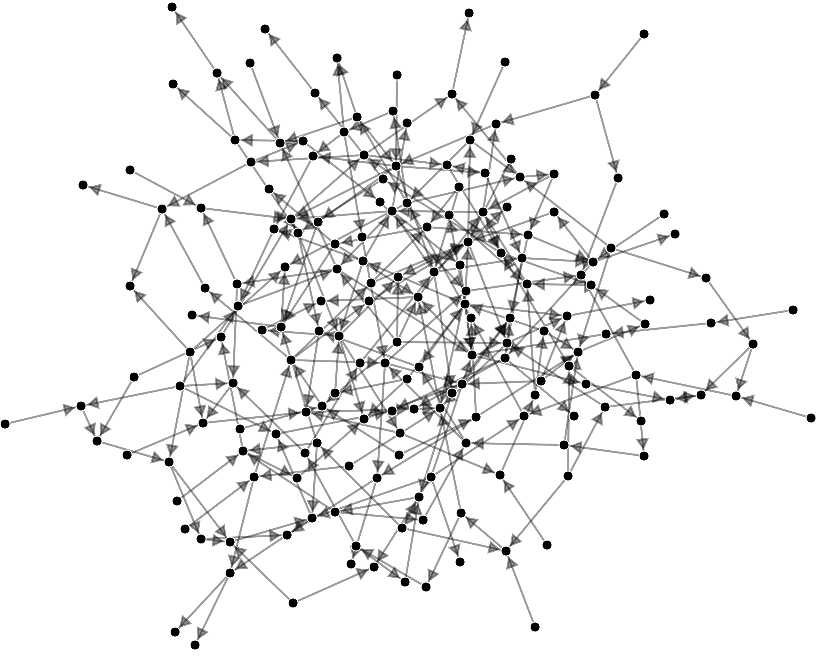}
    \label{fi:large-Random-ED}
	} 
   \subfigure[]{
    \includegraphics[width=\mypicsize] {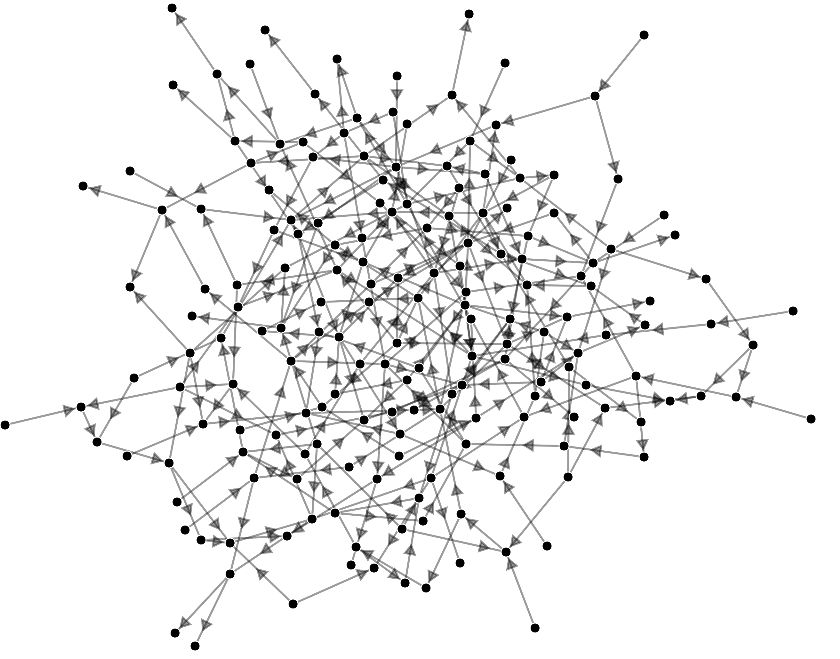}
    \label{fi:large-Random-EX}
	} 	
\caption{\label{fi:large-example-Planar-Random}(a)\&(b) Drawings of a \planar digraph with 100 vertices and 250 edges; (c)\&(d) Drawings of a \random digraph with 200 vertices and 280 edges. In the drawings on the left side the arrows are placed by a common editor; in the drawings on the right side the arrows are placed by our exact technique.}
\end{figure}

%\begin{figure}[]
%  \centering
%    \subfigure{
%    \includegraphics[width=\mypicsize] {large-drawings/200-14-RAND-ED.png}
%    \label{fi:large-Random-ED}
%	} 
%   \subfigure{
%    \includegraphics[width=\mypicsize] {large-drawings/200-14-RAND-EX.png}
%    \label{fi:large-Random-EX}
%	} 
%\caption{\label{fi:large-example-Random}Drawings of a \random digraph with 200 vertices and 280 edges. (left) The arrows are placed by a common editor; (right) The arrows are placed by our exact technique.}
%\end{figure}

\section*{Appendix B -- Additional Charts from the Experiments}

We report some additional charts from the experiments of Section~\ref{se:experiments}; see Figs.~\ref{fi:supporting-charts-1} and~\ref{fi:supporting-charts-2}. Moreover, in order to evaluate the scalability of our approach, we extended both \planar and \random sets with $30$ larger instances each ($6$ graphs for each number of vertices $n \in \{600, 700, \dots, 1000\}$) and ran our algorithms on them. 
The behavior of the algorithms is similar to that reported for smaller instances: the quality of \heurglobal's solutions almost coincides with the quality of \opt's solutions; see Figs.~\ref{fi:large-pl-conflicts-A},~\ref{fi:large-rand-conflicts-A}. Constructing $C_A$ remains the most expensive procedure, especially for the larger instances in the \planar set; see Figs.~\ref{fi:large-pl-solvingTime-A},~\ref{fi:large-pl-solvingTime-ConflictGraph-A}. Finally, our algorithms still generate significatively less crossings than \editor; see Figs.~\ref{fi:large-pl-crossings-A},~\ref{fi:large-rand-crossings-A}. Since \editor can use positions that are not valid for our algorithms, its performance in terms of number of overlaps tends to compare with the performances of our heuristics on some of the most cluttered layouts (see, e.g., the largest instances of \random). 

\begin{figure}[!h]
  \centering
  \subfigure[]{
    \includegraphics[width=0.46\textwidth] {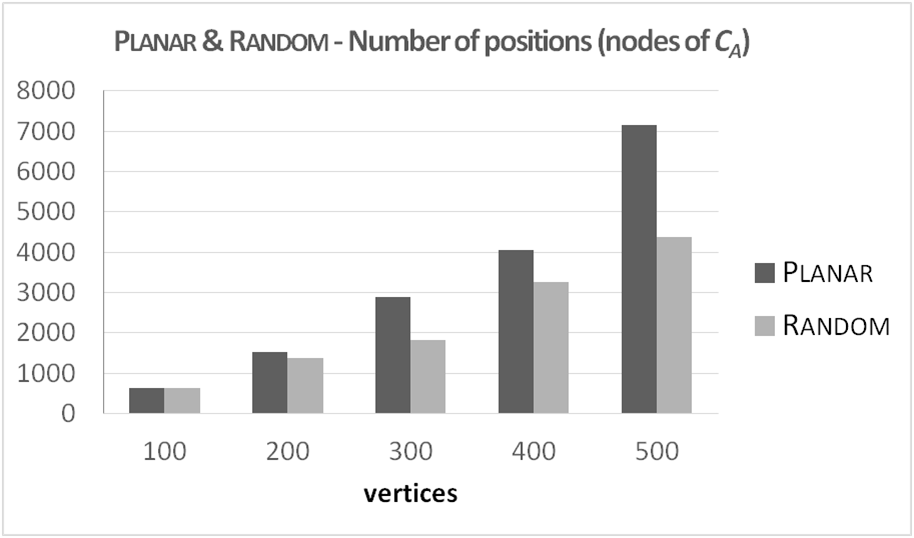}
    \label{fi:planar-random-positions}
	} 
    \subfigure[]{
    \includegraphics[width=0.46\textwidth] {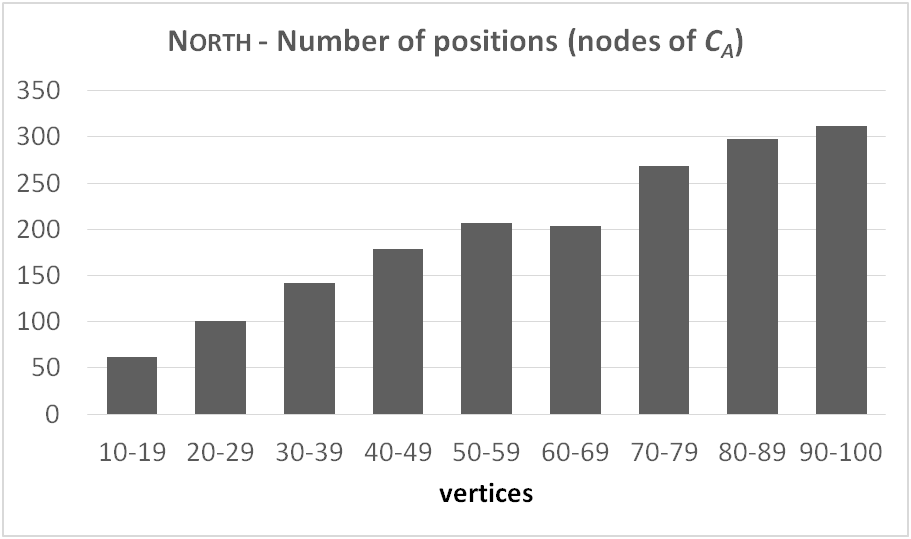}
    \label{fi:north-positions}
	} 
    \subfigure[]{
    \includegraphics[width=0.46\textwidth] {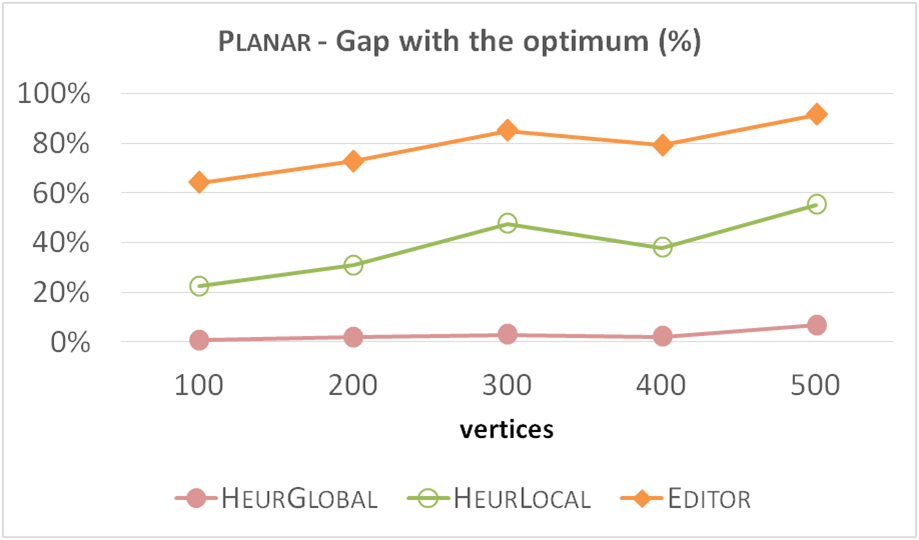}
    \label{fi:pl-conflicts-gap-A}
    }
  	\subfigure[]{
    \includegraphics[width=0.46\textwidth] {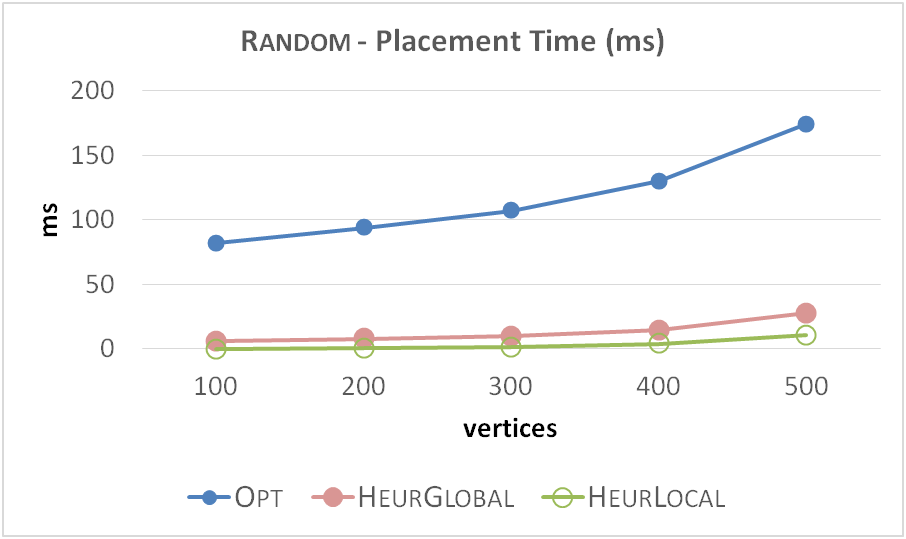}
    \label{fi:rand-solvingTime-A}
	} 
    \subfigure[]{
    \includegraphics[width=0.46\textwidth] {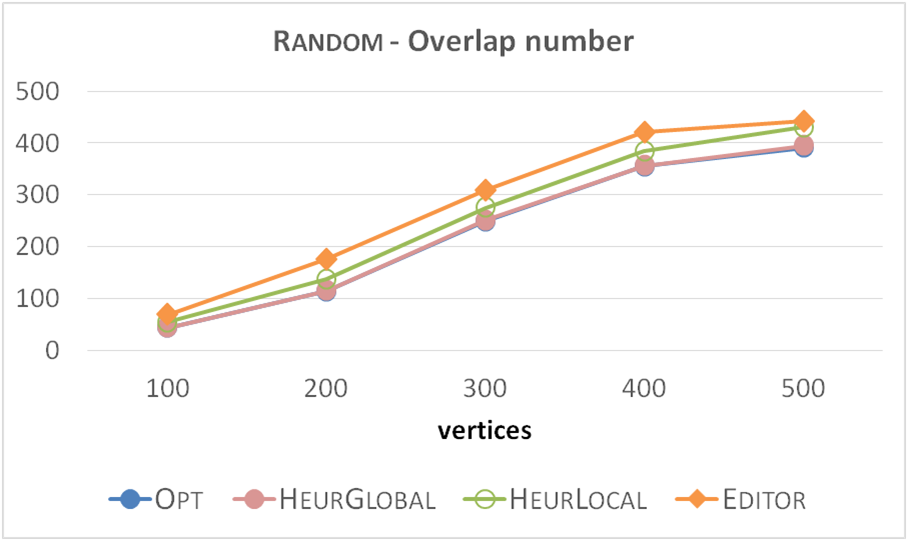}
    \label{fi:rand-conflicts-A}
	} 
   \subfigure[]{
    \includegraphics[width=0.46\textwidth] {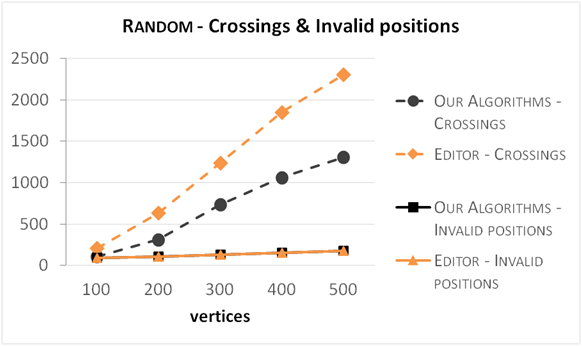}
    \label{fi:rand-crossings-A}
	} 
\caption{Number of positions (nodes in $C_A$) for: 
\planar and \random (a), and \north (b). --- 
\planar: (c) number of overlaps, relative to \opt. --- 
\random: (d) placement time; 
(e) number of overlaps; 
(f) number of crossings (between an edge/vertex and an arrow) and of invalid positions.}
\label{fi:supporting-charts-1}
\end{figure}

\begin{figure}[]
  \centering
	\subfigure[]{
    \includegraphics[width=0.46\textwidth] {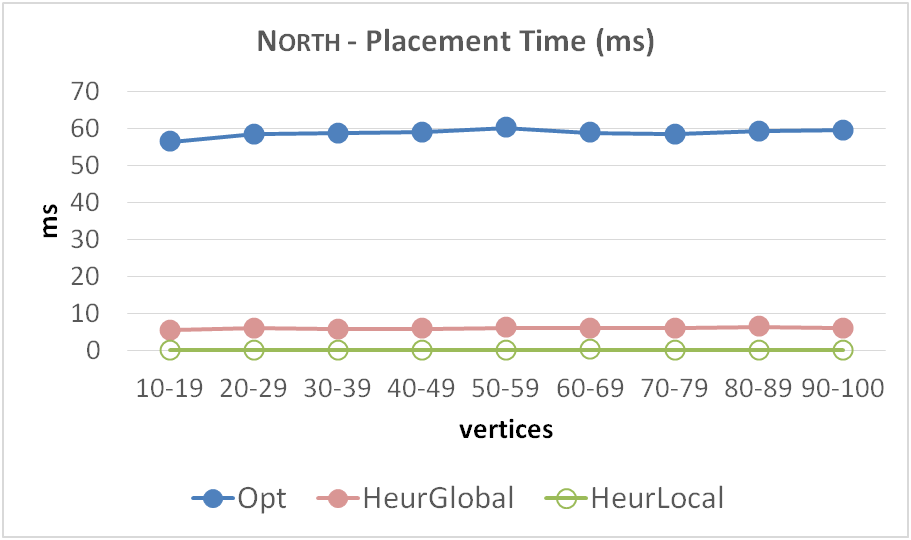}
    \label{fi:north-solvingTime-A}
	} 
    \subfigure[]{
    \includegraphics[width=0.46\textwidth] {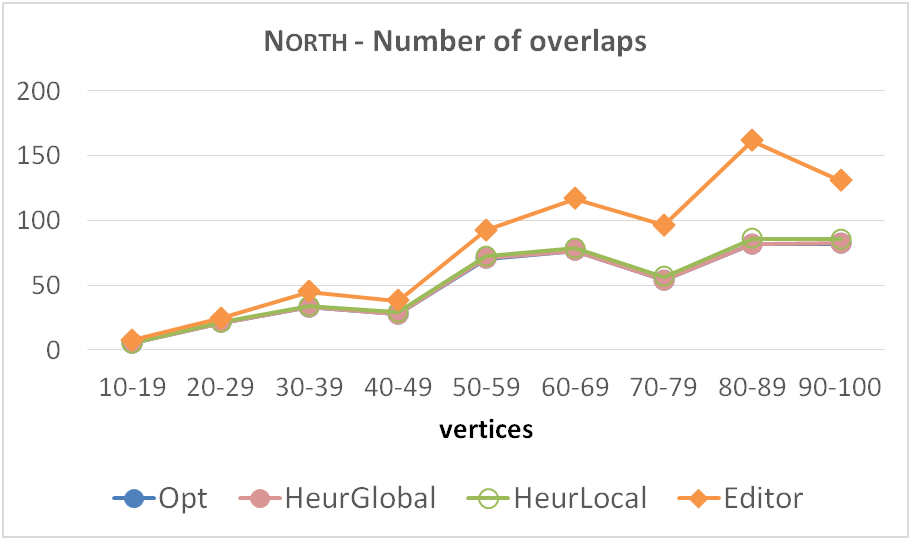}
    \label{fi:north-conflicts-A}
	} 
   \subfigure[]{
    \includegraphics[width=0.46\textwidth] {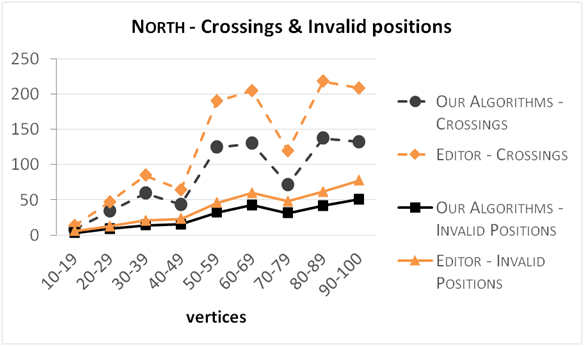}
    \label{fi:north-crossings-A}
	} 
\caption{\north: (a) placement time; 
(b) number of overlaps; 
(c) number of crossings (between an edge/vertex and an arrow) and of invalid positions.}
\label{fi:supporting-charts-2}
\end{figure}

\begin{figure}[]
  \centering
  	\subfigure[]{
    \includegraphics[width=\mypicsize] {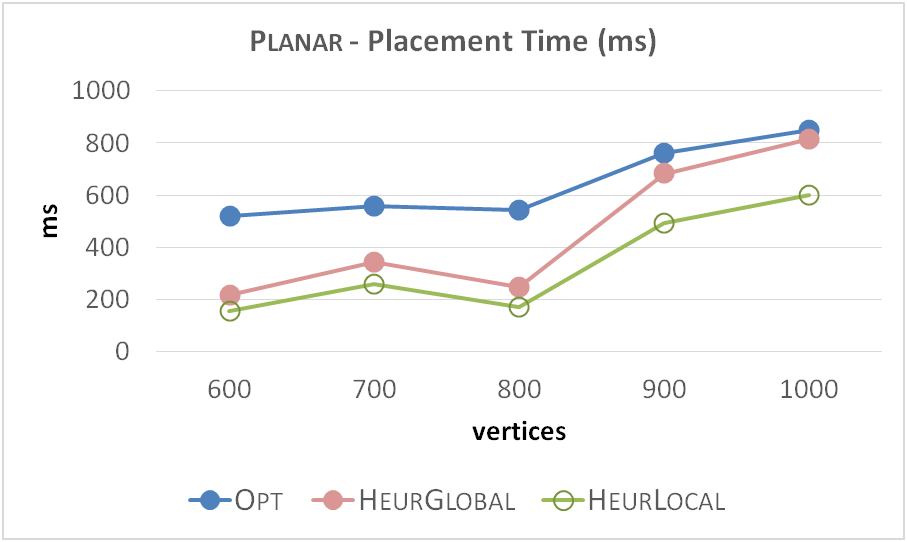}
    \label{fi:large-pl-solvingTime-A}
	} 
    \subfigure[]{
    \includegraphics[width=\mypicsize] {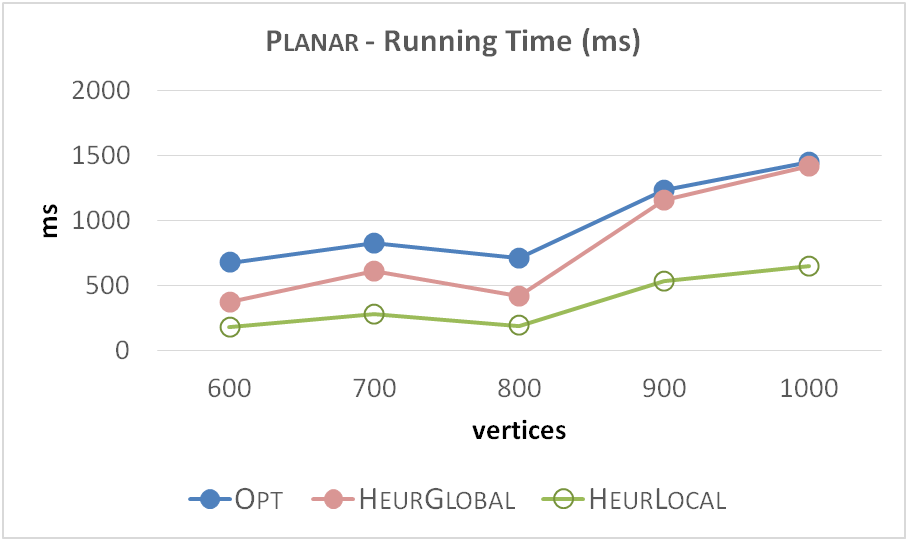}
    \label{fi:large-pl-solvingTime-ConflictGraph-A}
	} 
    \subfigure[]{
    \includegraphics[width=\mypicsize] {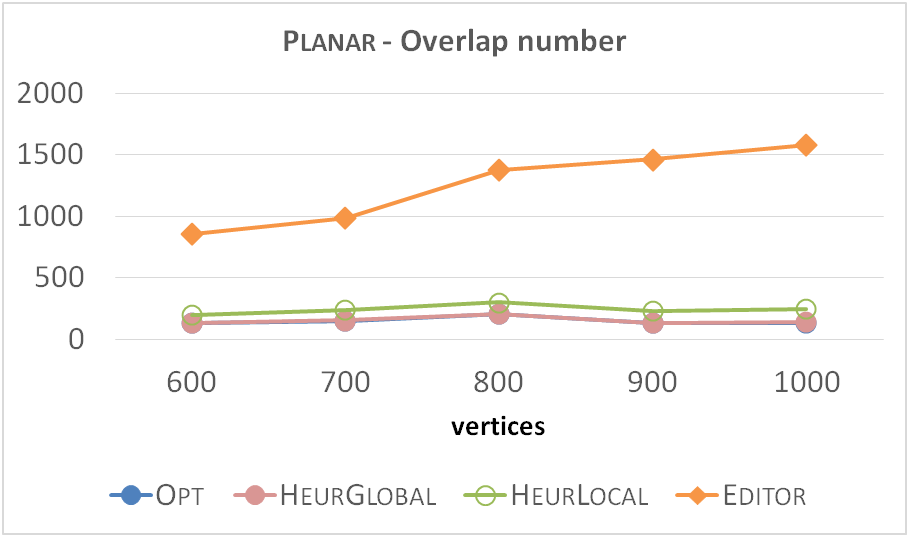}
    \label{fi:large-pl-conflicts-A}
	} 
    \subfigure[]{
    \includegraphics[width=\mypicsize] {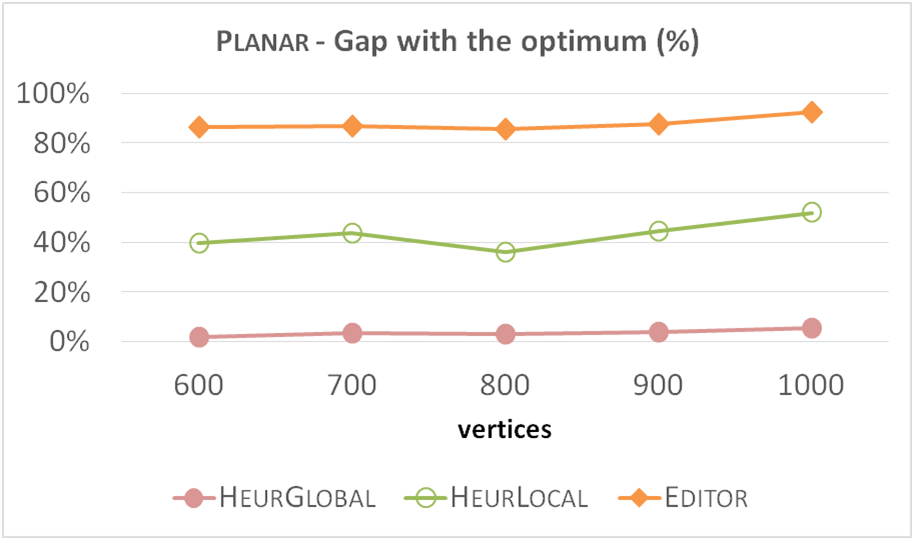}
    \label{fi:large-pl-conflicts-gap-A}
	} 
   \subfigure[]{
    \includegraphics[width=\mypicsize] {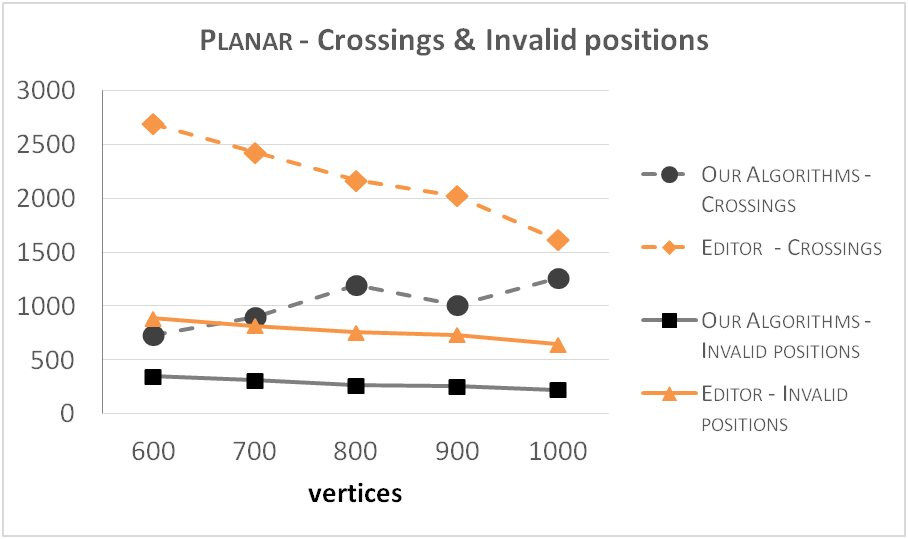}
    \label{fi:large-pl-crossings-A}
	} 
\caption{Larger graphs of \planar: (a) placement time;
(b) total running time;
(c) number of overlaps;
(d) number of overlaps, relative to \opt;
(e) number of crossings (between an edge/vertex and an arrow) and of invalid positions.}
\label{fi:large-planar}
\end{figure}

\begin{figure}[]
  \centering
  	\subfigure[]{
    \includegraphics[width=\mypicsize] {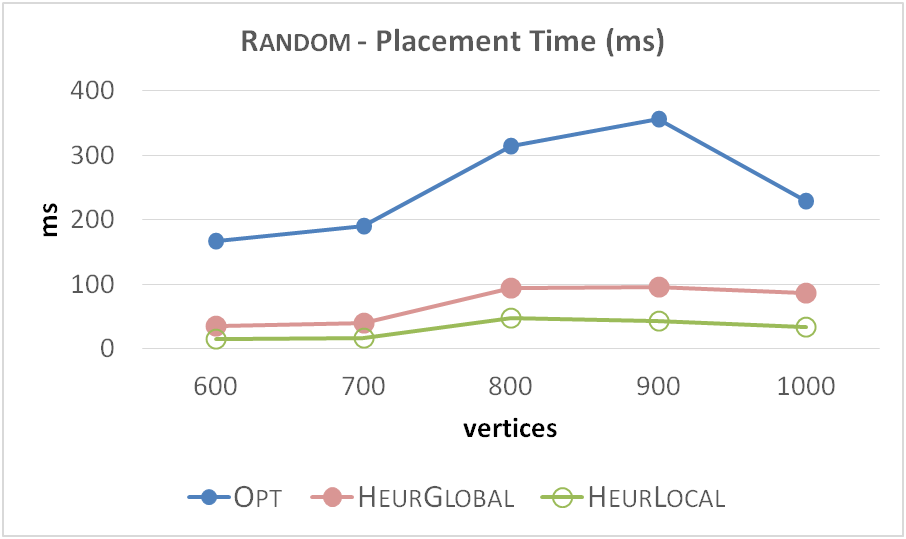}
    \label{fi:large-rand-solvingTime-A}
	} 
    \subfigure[]{
    \includegraphics[width=\mypicsize] {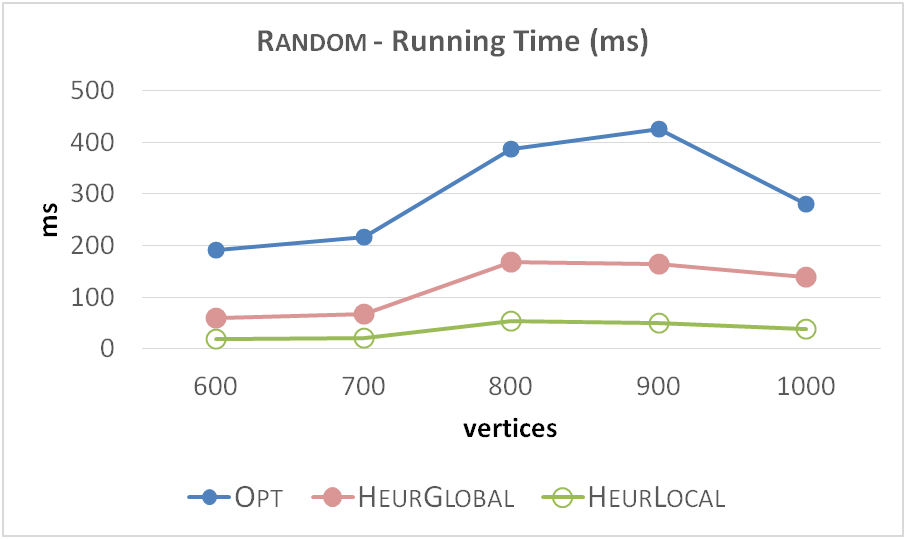}
    \label{fi:large-rand-solvingTime-ConflictGraph-A}
	} 
    \subfigure[]{
    \includegraphics[width=\mypicsize] {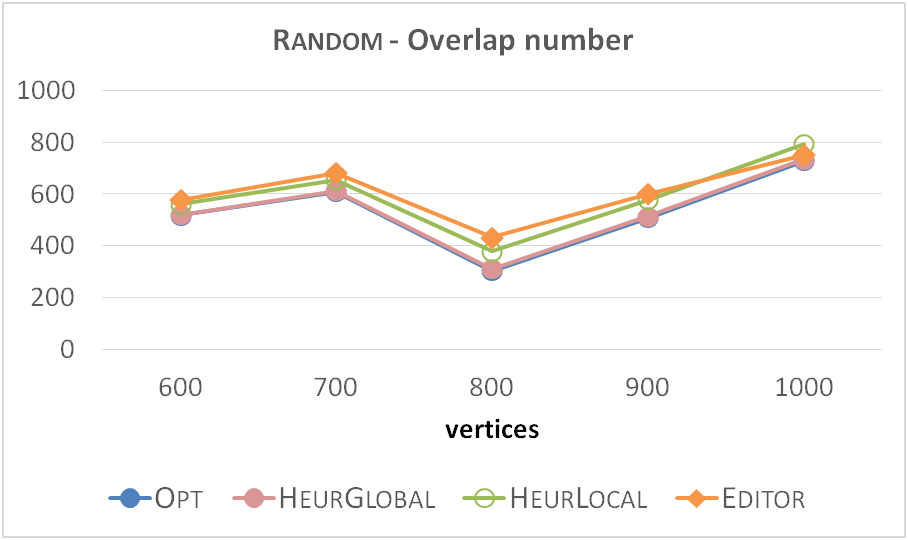}
    \label{fi:large-rand-conflicts-A}
	} 
    \subfigure[]{
    \includegraphics[width=\mypicsize] {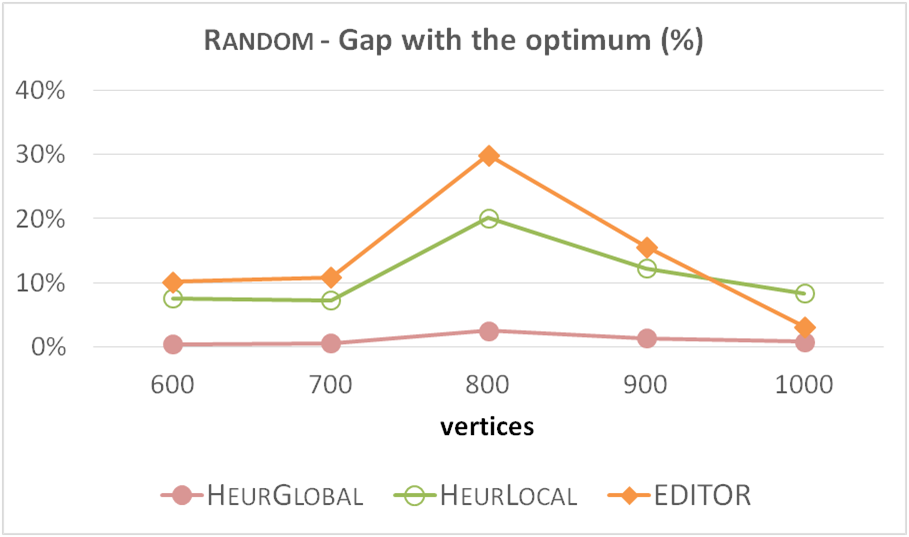}
    \label{fi:large-rand-conflicts-gap-A}
	} 
   \subfigure[]{
    \includegraphics[width=\mypicsize] {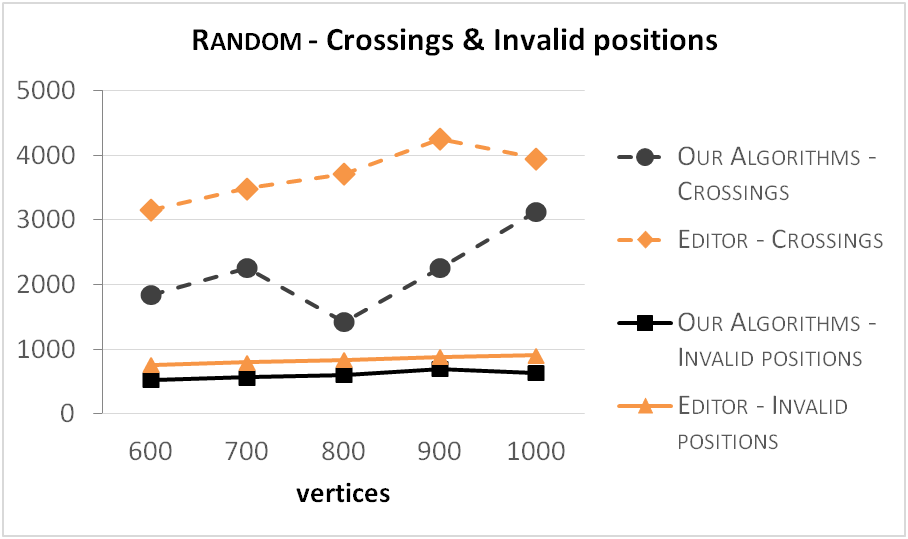}
    \label{fi:large-rand-crossings-A}
	} 
\caption{Larger graphs of \random : (a) placement time;
(b) total running time;
(c) number of overlaps;
(d) number of overlaps, relative to \opt;
(e) number of crossings (between an edge/vertex and an arrow) and of invalid positions.}
\label{fi:large-random}
\end{figure}

\clearpage

\section*{Appendix C -- NP-hardness Proof}

\setcounter{theorem}{0}
\begin{theorem}
The \ap problem is NP-hard.
\end{theorem}
\begin{proof}
We show that \ap is NP-hard by a reduction from {\sc Planar 3-SAT} ({\sc P3SAT})~\cite{DBLP:journals/siamcomp/Lichtenstein82}, the special case of {\sc 3-SAT} in which the bipartite graph of variables and clauses is planar. Our proof follows the proof of Knuth and Raghunathan~\footnote{Knuth, D.E., Raghunathan, A.: The problem of compatible representatives. SIAM J. Discrete Math. $5(3)$, $422$-$427$ ($1992$)} for the Metafont labeling problem in which the variables are placed onto a horizontal line, while the three-legged clauses are placed above or below them. The clauses are properly nested so that none of the legs between clauses and variables cross each other. The connections can either be positive or negative. We remark that our reduction technique is similar to those used in the context of edge and point labeling: labeling points with circular labels~\cite{DBLP:journals/ijcga/StrijkW01} and with axis-parallel rectangular labels in different models~\cite{Marks91thecomputational},~\cite{DBLP:journals/comgeo/KreveldSW99}; labeling edges with axis-parallel rectangular labels~\cite{Wolff00asimple}. To reduce {\sc P3SAT} to \ap, we describe how to construct from any planar $3$-CNF formula a (possibly) connected drawing $\Gamma$ such that there exists a valid placement $P_\Gamma$ of the arrows with $\ov{P_\Gamma}=0$ if and only if the formula is satisfiable.
Let $\phi$ be an instance of the {\sc P3SAT} problem consisting of the set of clauses $c_1,\dots,c_m$, each having three literals from the Boolean variables $x_1,\dots,x_n$. We now explain how to construct a variable gadget for each variable $x_i$ and a clause gadget for each clause $c_i$ of $\phi$. We use two building blocks. 

The first building block, which we call \emph{triangle-block}, is shown in Fig.~\ref{fi:triangle-block}: it is composed of an equilateral triangle $T$ having side length equal to $2r_E\sqrt{3}$. By basic geometry, it follows that along each edge $e$ of $T$ a circle $C_e$ (modeling arrow $a_e$) with its center on $e$ and radius $r_E$ admits only two alternative valid positions represented in the figure by circles with solid and dashed boundaries, respectively. In the following, we simply call a valid position represented by a circle with solid (dashed) boundary a \emph{solid (dashed) position}. 
%Note that the position of solid and dashed labels along the edges of $T$ can be inverted. 
It can be observed that there exist only two valid placements for $T$, one called \emph{solid placement} that uses only solid positions and one called \emph{dashed placement} that uses only dashed positions. We call \emph{base} of $T$, the horizontal edge of $T$, while we call the other two edges of $T$ the \emph{left} and the \emph{right} edges of $T$, respectively. 

The second building block which we call \emph{trapezoid-block} is a right trapezoid $Q$ in which the vertical side is missing. We call the two horizontal edges of $Q$ the \emph{major base} and the \emph{minor base} of $Q$ and the third edge of $Q$ the \emph{diagonal side} of $Q$. As shown in Fig.~\ref{fi:trapezoid-block}, we have two variants of $Q$ (called \emph{$Q$-left} and \emph{$Q$-right}) that are symmetric. The lengths of the three edges of $Q$ are chosen as follows. Denote by $q_1$ and $q_2$ ($q_3$ and $q_5$) the two end-points of the minor (major) base of $Q$ and denote by $q_4$ the projection of $q_1$ along the major base. Note that $\overline{q_1q_4}$ is the \emph{height} of $Q$. We have the following relations: $\overline{q_3q_4}=\overline{q_1q_4}$, $\overline{q_1q_2}=\overline{q_4q_5}$. The minor base of $Q$ has length $2r_E\sqrt{3}$, the height of $Q$ has length $r_E (1+ \frac{\sqrt{2}}{2})$, then the major base of $Q$ has length $r_E (1+ \frac{\sqrt{2}}{2}) + 2r_E\sqrt{3}$. Observe that the height of $Q$ is less than $2r_E$ and that the diagonal side of $Q$ is the diagonal of a square with side $\overline{q_3q_4}$.
By basic geometry it follows that: (i) the diagonal side of $Q$ admits a unique valid position that is tangent to segment $\overline{q_3q_4}$; (ii) the minor and major bases of $Q$ admit only two alternative valid positions. Also in this case it can be observed that beside the unique valid position of the diagonal side, the selected valid positions of the minor and major bases of $Q$ must be either both solid or both dashed. A valid placement of $Q$ using the unique valid position of the diagonal side and the two solid (dashed) positions of the minor and major bases is called a \emph{solid (dashed) placement}.

\smallskip

\begin{figure}[tb]%[htp]
\centering
\subfigure[]{\label{fi:triangle-block}\includegraphics[scale=0.65, page=1]{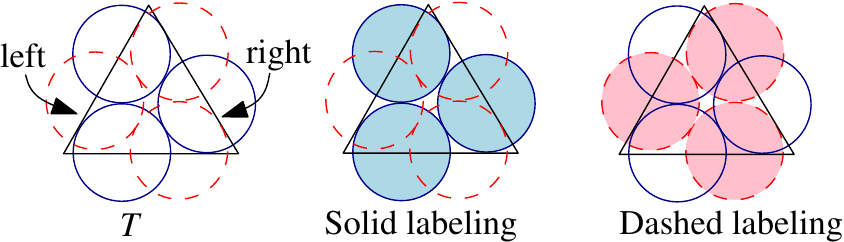}}
\hfil
\subfigure[]{\label{fi:trapezoid-block}\includegraphics[scale=0.65, page=3]{gadget}}
 \caption{\small{(a) A triangle-block $T$ and the two possible valid placements of $T$: the solid placement and the dashed placement. (b) The two variants $Q$-left and $Q$-right of a trapezoid-block $Q$. The diagonal side admits a unique valid position, that is represented by a green circle.}}\label{fi:building-blocks}
\end{figure}

\noindent{\emph{Variable gadget.}} The variable gadget is composed by a horizontal chain $\mathcal{T}=\lbrace T_1,\dots,T_k \rbrace$ of triangle-blocks, for a given odd $k \geq 5$. Each triangle block $T_i$ with $i$ even is an upside-down triangle that shares its left edge with the right edge of the triangle $T_{i-1}$ and its right edge with the left edge of the triangle $T_{i+1}$ as shown in Fig.~\ref{fi:variable-gadget}. This implies that the choice of one between the solid placement or the dashed placement for the leftmost triangle $T_1$ enforces the same kind of placement for all other triangle-blocks of the chain. Hence, we call a \emph{solid (dashed) placement} of $\mathcal T$, a placement of $\mathcal T$ composed of only solid (dashed) positions. In our construction, the truth (false) value of a variable $x_i$ is encoded by a solid (dashed) placement of the corresponding variable gadget $\mathcal T^i$.

\begin{figure}[tb]%[htp]
\centering
\subfigure[]{\label{fi:variable-gadget}\includegraphics[scale=0.6, page=2]{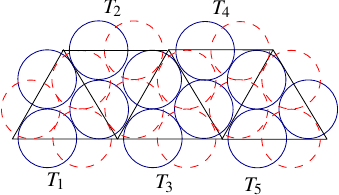}}
\hfil
\subfigure[]{\label{fi:leg-gadget}\includegraphics[scale=0.6, page=4]{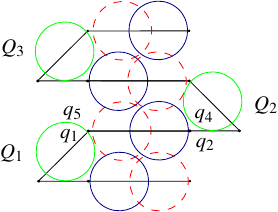}}
 \caption{\small{(a) A variable gadget $\mathcal{T}=\lbrace T_1,\dots,T_5 \rbrace$. (b) A chain $\mathcal{Q}=\{ Q_1,Q_2,Q_3 \}$ of trapezoid-blocks, which represents the leg of a clause gadget. As example, the two end-points $q_1$ and $q_2$ of the minor base of $Q_1$ coincide with the points $q_5$ and $q_4$ of the major base of $Q_2$.}}\label{fi:gadgets}
\end{figure}

\smallskip

\noindent{\emph{Clause gadget.}} The clause gadget is composed of three \emph{vertical legs} and one \emph{horizontal part} (refer to Fig.~\ref{fi:clause-gadget}). The horizontal part is composed of two horizontal chains $\mathcal{T}_{L}= \lbrace T_{1},\dots,T_{l} \rbrace$ and $\mathcal{T}_{R}= \lbrace T_{1},\dots,T_{r} \rbrace$ of triangle-blocks (see Fig.~\ref{fi:clause-gadget}). We have the following properties: (i) The solid and dashed positions along the edges of the triangle-blocks of $\mathcal{T}_{R}$ are inverted with respect to those of the triangle-blocks of $\mathcal{T}_{L}$; (ii) For every odd $i$, the base of $T_{i} \in \mathcal{T}_{L}$ and of $T_{i} \in \mathcal{T}_{R}$ lie on the same horizontal line $l_1$; (iii) The solid position of the right edge of $T_{l} \in \mathcal{T}_{L}$ is tangent to the solid position of the left edge of $T_{r} \in \mathcal{T}_{R}$. Denote by $p$ this tangent point.

\begin{figure}[tb]%[h!]
\begin{center}
\includegraphics[scale=0.6, page=5]{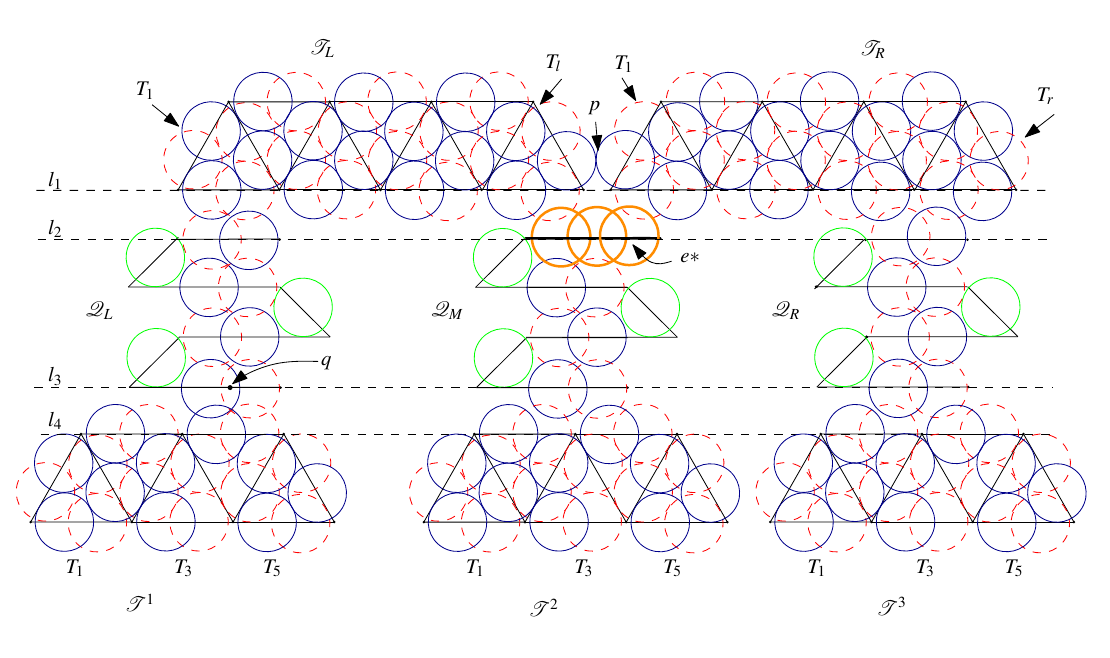}
\end{center}
\caption{Illustration of a clause gadget. The clause of a boolean formula $\phi$ modeled by this example is $c_1=(\overline{x_1}\vee x_2 \vee x_3$). Note that if $x_1$ is true and $x_2$ and $x_3$ are false, then edge $e^*$ does not admit a valid position.}
\label{fi:clause-gadget} 
\end{figure} 

We now describe the three legs, called \emph{left leg}, \emph{middle leg}, and \emph{right leg}. 
Each leg is a vertical chain $\mathcal{Q}=\{Q_1,\dots,Q_k\}$ (with $k$ odd) of trapezoid-blocks composed in the following way (see Fig.~\ref{fi:leg-gadget}): each $Q_i$ with $i$ even is a $Q$-right trapezoid-block that shares its major base with the minor base of $Q_{i-1}$ (that is a $Q$-left trapezoid). We denote by $\mathcal{Q}_{L}$, $\mathcal{Q}_{M}$ and $\mathcal{Q}_{R}$ the vertical chains of the left, the middle and the right legs, respectively. We denote by the \emph{top base (bottom base)} of each leg the minor base (major base) of the topmost (bottom most) trapezoid-block $Q_k$ ($Q_1$) of each leg. We have the following properties: (i) $\mathcal{Q}_{L}$, $\mathcal{Q}_{M}$, $\mathcal{Q}_{R}$ have the same number $k$ of trapezoid-blocks, and therefore, their top bases (bottom bases) lie on the same horizontal line $l_2$ ($l_3$); (ii) The top base of $\mathcal{Q}_{M}$ is longer than those of $\mathcal{Q}_{L}$ and $\mathcal{Q}_{R}$ so that this edge (denoted by $e*$) is the only one that admits more than two valid positions. We set the length of $e*$ equal to $r_E(4\sqrt{3}-2)$. 

We place the legs below the horizontal part in the following way: we set the Euclidean distance $d$ between $l_1$ and $l_2$ equal to $r_E \left( 1+ \frac{\sqrt{2}}{2} \right)$, that is $d < 2r_E$ and for the left leg, we vertically align the middle point of the top base of $\mathcal{Q}_{L}$ with the middle point of the base of $T_1 \in \mathcal{T}_{L}$; for the middle leg we vertically align the middle point of the top base of $\mathcal{Q}_{M}$ with point $p$; for the right leg we vertically align the middle point of the top base of $\mathcal{Q}_{R}$ with the left end-point of the base of $T_{r} \in \mathcal{T}_{R}$.

The above geometric relations guarantee the following properties:
\textbf{(a)} The choice of one between the solid position and the dashed position for the bottom-base of $\mathcal{Q}_{L}$ ($\mathcal{Q}_{R}$) enforces the same kind of positions for all the other trapezoid-blocks of the leg and for all the triangle-blocks of $\mathcal{T}_{L}$ ($\mathcal{T}_{R}$). The same holds for the trapezoid-blocks of $\mathcal{Q}_{M}$. That is, the truth or the false value is transmitted along each leg and along the horizontal part of the clause.
\textbf{(b)} If each leg transmits the false value (i.e. the enforced positions are dashed), then the length of edge $e*$ guarantees that there does not exist a valid position for $e*$. Otherwise, if at least one of the three legs transmits the truth value (i.e. at least one leg uses solid positions), then there exists a valid position for edge $e*$.

In Fig.~\ref{fi:clause-gadget} we also show how each clause leg is connected to a variable gadget. We align the bases of the upside-down triangles of each variable gadget $\mathcal T^i$ (corresponding to a variable $x_i$) along a horizontal line $l_4$ and, as in the case of $l_1$ and $l_2$, we set the Euclidean distance between $l_3$ and $l_4$ equal to $d=r_E \left( 1+ \frac{\sqrt{2}}{2} \right)$. Depending on the literal we attach the clause leg to the corresponding variable gadget in one of two possible places. Denote by $q$ the middle point between the centers of the solid and dashed positions of the bottom base of a leg. If the literal is negated we vertically align $q$ with the middle point of the base of a triangle-block $T_i$ with $i$ even (i.e. an upside-down triangle) of the variable gadget $\mathcal T^i$. If the literal is not negated, we vertically align $q$ with the end-point shared by the bases of the two triangle blocks $T_i$ and $T_{i+2}$ with $i$ even. 

\begin{figure}[bt]%[h!]
\begin{center}
\includegraphics[scale=0.5, page=6]{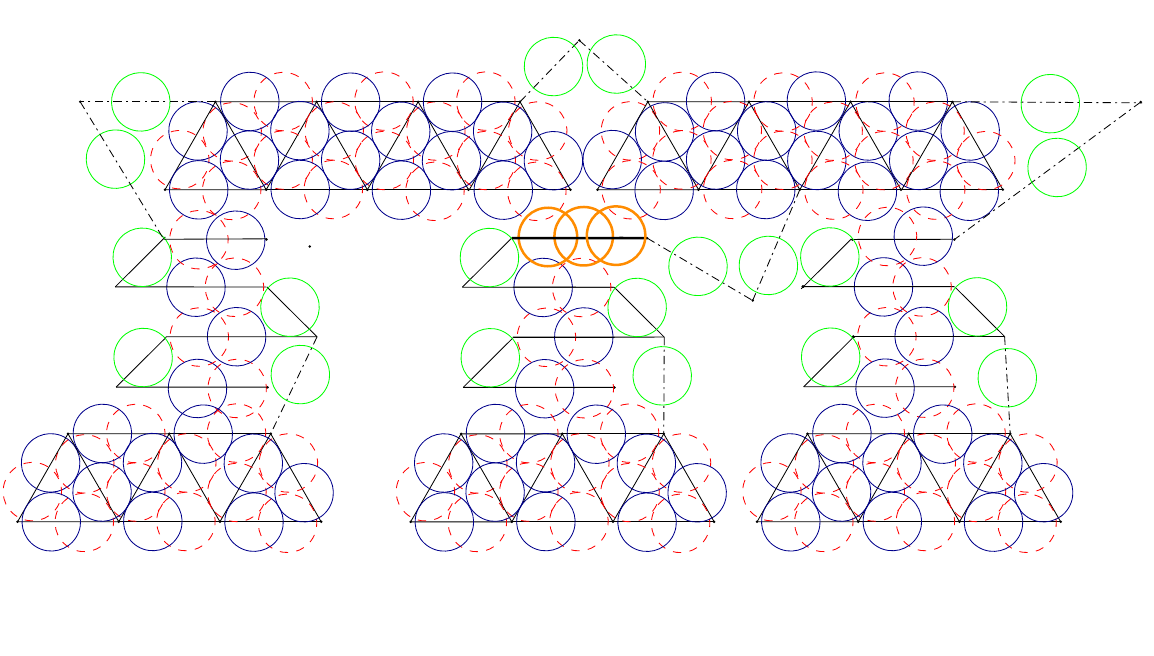}
\end{center}
\caption{The clause of Fig.~\ref{fi:clause-gadget} where some other edges have been added to make $\Gamma$ connected. The added edges are represented by dashed dotted edges and admit always a valid position.}
\label{fi:clause-connected} 
\end{figure} 

If all literals of a clause are false, then the edge $e*$ does not admit a valid position. Indeed all the three legs of the clause transmits the false value (i.e. the selected positions of all the three legs are dashed). Otherwise if at least one literal is true, then $e*$ and all the other edges admit a valid position. Indeed in this case, the clause leg attached to the true literal transmits the truth value (i.e. the selected positions for it are solid). Therefore, given any  planar $3$-CNF formula $\phi$ there exists a valid placement $P_\Gamma$ of the arrows such that $\ov{P_\Gamma}=0$ if and only if $\phi$ is satisfiable. 

Moreover the construction guarantees that the variables are placed onto a horizontal line and they can be extended to reach all the clauses. The same holds for the three legs and the two horizontal parts of the clauses.

Observe that the drawing $\Gamma$ described so far is not connected. In order to make $\Gamma$ connected, we can add a suitable number of extra edges, represented in Fig.~\ref{fi:clause-connected} as dashed-dotted segments, that always admit a valid position.\qed
\end{proof}

%\clearpage

\end{document}